\documentclass[12pt,A4,twoside]{preprint}
\usepackage{amsmath}
\usepackage{amsfonts}
\usepackage{graphicx}
\usepackage{makeidx}
\usepackage{times}
\usepackage{amsthm}
\usepackage{color}
\usepackage{mhequ}
\usepackage{dsfont}
\usepackage[scanall]{psfrag}
\usepackage{epsfig}
\usepackage[margin=2.5cm]{geometry}

\usepackage{url}

\def\thecomma{\ifx,\thenext \else\ifx;\thenext \else\ifx.\thenext
\else\ifx!\thenext \else\ifx:\thenext\else\ifx)\thenext \else \
\fi\fi\fi\fi\fi\fi}
\def\condblank{\futurelet\thenext\thecomma}
\def\ie{{\it i.e.}\condblank}

\def\fref#1{Fig.~\ref{#1}}

\def\cref#1{Condition~\ref{#1}}
\def\Cref#1{Corollary~\ref{#1}}
\def\eref#1{(\ref{#1})}
\def\sref#1{Sect.~\ref{#1}}
\def\lref#1{Lemma~\ref{#1}}

\def\tref#1{Theorem~\ref{#1}}
\def\dref#1{Definition~\ref{#1}}
\def\pref#1{Proposition~\ref{#1}}
\numberwithin{equation}{section}

\newenvironment{myitem}
{\begin{itemize}
  \setlength{\itemsep}{1pt}
  \setlength{\parskip}{0pt}
  \setlength{\parsep}{0pt}}
{\end{itemize}}

\newenvironment{myenum}
{\begin{enumerate}
  \setlength{\itemsep}{1pt}
  \setlength{\parskip}{0pt}
  \setlength{\parsep}{0pt}}
{\end{enumerate}}

\newtheorem*{oneshot}{Hypothesis~\ref{h:assumption}}
\newtheorem{theorem}{Theorem}[section]
\newtheorem{lemma}[theorem]{Lemma}
\newtheorem{proposition}[theorem]{Proposition}
\newtheorem{definition}[theorem]{Definition}

\newtheorem{hypothesis}[theorem]{Hypothesis}

\newtheorem{remark}[theorem]{Remark}

\newtheorem{corollary}[theorem]{Corollary}

\def\II{\mathcal{I}}
\def\KK{\mathcal{K}}
\def\MM{\mathcal{M}}

\def\NN{\mathcal{N}}

\def\TT{\mathcal{T}}
\def\OO{\mathcal{O}}
\def\EE{\mathcal{E}}
\def\KK{\mathcal{K}}
\def\LL{\mathcal{L}}
\def\YY{\mathcal{Y}}
\def\L{{\rm{L}}}
\def\R{{\rm{R}}}
\def\S{{\bf{S}}}
\let\phi=\varphi
\let\rho=\varrho
\def\i{{\rm i}}
\def\s{{\rm s}}
\def\ss{{\bf s}}

\def\emx#1{\emph{#1}\index{#1}}

\makeindex
\addtocontents{toc}{\protect\setcounter{tocdepth}{2}}
\begin{document}
\title{Trees of nuclei and bounds\\ on the number of triangulations of
  the 3-ball}
\author{P.~Collet${}^1$, J.-P.~Eckmann${}^{2,3}$ and M.~Younan${}^2$}
\institute{${}^1$Centre de Physique Th\'eorique, CNRS UMR 7644,\\ Ecole
  Polytechnique, F-91128 Palaiseau Cedex (France)\\
${}^2$D\'epartement de Physique Th\'eorique,
Universit\'e de Gen\`eve, CH-1211 Gen\`eve 4 (Switzerland)\\
${}^3$Section de Math\'ematiques,
Universit\'e de Gen\`eve, CH-1211 Gen\`eve 4 (Switzerland)}
\maketitle
\begin{abstract}
  Based on the work of Durhuus-J{\'o}nsson and Benedetti-Ziegler, we
  revisit the question of the number of triangulations of the
  3-ball. We introduce a notion of nucleus (a triangulation of the
  3-ball without internal nodes, and with each internal face having at
  most 1 external edge). We show that every triangulation can be built
  from trees of nuclei. This leads to a new reformulation of Gromov's
  question: We show that if the number of rooted nuclei with $t$
  tetrahedra has a bound of the form $C^t$, then the number of rooted
  triangulations with $t$ tetrahedra is bounded by $C_*^t$.
\end{abstract}
\tableofcontents

\section{Introduction}
In this paper, we study the question of the number of triangulations of
the 3-ball by tetrahedra. The case of the 2-ball was exactly solved by
Tutte in \cite{Tutte1962}. He showed in particular that the number of rooted
triangulations of the 2-sphere with $N$ vertices is
$\OO(1)\,N^{-5/2}(256/27)^{N}$. It is natural to ask if analogous bounds
are true in higher dimension. Such results could have applications in
models of Statistical Mechanics (foams \cite{Magnasco1992}, quantum gravity
\cite{ADJ1997}, or  glassy dynamics \cite{Proglass2007,astesherrington1999,eckmannglass2007,JPY2011}) where the
exponential rate of growth can be interpreted as an entropy.  In
\cite{gromov2000}, Gromov asked whether the number of triangulations of the
3-sphere is bounded by $C^N$ for some constant $C$ when there are $N$
tetrahedra (facets) in the triangulation. To date, this question remains
open. However Pfeifle and Ziegler proved in \cite{PZ2004} a super
exponential lower bound for the number of triangulations of the 3-ball
as a function of the number of vertices. This does not answer
negatively Gromov's question (which is in terms of the number of tetrahedra)
but makes the problem of proving an exponential bound in terms of the
number of tetrahedra even more challenging.

There are several studies in the direction of answering
the question, which we summarize now. 
In \cite{DJ1995}, Durhuus and
J{\'o}nsson gave the construction of a class of triangulations for which
they could show a bound of the form $C^N$. These triangulations are
obtained by building a tree of tetrahedra, which is obtained by
starting from a root tetrahedron and attaching tetrahedra to its
faces, and then attaching further tetrahedra to the new open
faces. Each tetrahedron is attached to the tree with just one
face. It is a common feature of tree-like constructions that they lead
to bounds of the form $C^N$: The prime example in our context is of
course the celebrated work of Tutte \cite{Tutte1962} mentioned above.
Coming back to Durhuus and J{\'o}nsson, once the tree is constructed, they
now collapse adjacent faces of the tree in such a way that at the end of the
procedure a triangulation of the 3-sphere is obtained. Their main
result says that the number of ways in which to do this is again
exponentially bounded. In this way, they construct a set of
triangulations of the 3-sphere with tetrahedra which is exponentially
bounded. They ask whether these are all possible triangulations.

In a later development, Benedetti and Ziegler \cite{BZ2011}, show that
the Durhuus and J{\'o}nsson construction, which they call ``locally
constructible'' (LC), does \emph{not} capture all triangulations of
the 3-sphere. Namely, they show that a 3-sphere with a 3-complicated
knot (made by tetrahedra) is not LC. They also carefully
discuss relations between LC and other classes of constructibility.

In the present paper, we define a larger class of triangulations, with
a construction similar to that of Durhuus and J{\'o}nsson, but which uses
more general basic elements than the simple tetrahedron, which we call
\emx{nuclei}. We prefer to work with 3-balls, and bounds on 3-spheres
can be obtained from a bound on triangulations of a tetrahedron. This
is usually done by removing a tetrahedron from the 3-sphere (see for
example \cite[Section 3]{BZ2011} ).

{\bf Nuclei} are defined as triangulations of the 3-ball with the following
special properties:
\begin{myenum}
  \item They have no internal nodes.
  \item Internal faces have at most \emph{one} external edge.
\end{myenum}
Obviously, the tetrahedron is a nucleus. The Furch-Bing ball
\cite{Furch1924}, \cite{Bing1964} and \cite{H2001Online} and the Bing 2-room house \cite{Bing1964} and \cite{H2001Online},
which are not nuclei,  can be reduced by our
procedure to one non-trivial nucleus, each.
The smallest non-trivial nucleus we know of, given in Table~\ref{table1}, has 12 nodes, and 37
tetrahedra, of which 17 have no external face. Nodes are numbered from 1 to 12, and Table~\ref{table1} gives
a list of the 37 tetrahedra.

\begin{table}[h]
\begin{center} 
\scalebox{0.8}{
\begin{tabular}{
|r r r r|r r r r|r r r r|r r r r|r r r r|}
\hline
{\bf 1}&{\bf 3}&{\bf 4}&10&
1&3&5&10&
{\bf 1}&{\bf 3}&5&{\bf 11}&
{\bf 1}&{\bf 4}&{\bf 6}&10&
1&5&7&8\\
1&{\bf 5}&{\bf 7}&{\bf 10}&
{\bf 1}&5&{\bf 8}&{\bf 11}&
{\bf 1}&{\bf 6}&7&{\bf 8}&
1&6&7&10&
{\bf 2}&3&{\bf 5}&{\bf 9}\\
2&3&5&11&
2&3&8&9&
2&{\bf 3}&{\bf 8}&{\bf 11}&
{\bf 2}&{\bf 5}&{\bf 6}&11&
{\bf 2}&{\bf 6}&11&{\bf 12}\\
{\bf 2}&{\bf 7}&{\bf 10}&11&
{\bf 2}&{\bf 7}&11&{\bf 12}&
{\bf 2}&8&{\bf 9}&{\bf 10}&
2&8&10&11&
3&4&9&10\\
{\bf 3}&{\bf 4}&9&{\bf 12}&
3&{\bf 5}&{\bf 9}&{\bf 10}&
{\bf 3}&{\bf 8}&9&{\bf 12}&
{\bf 4}&{\bf 5}&{\bf 6}&11&
{\bf 4}&{\bf 5}&{\bf 7}&8\\
4&5&8&11&
4&6&10&11&
4&7&8&9&
{\bf 4}&{\bf 7}&9&{\bf 12}&
4&8&9&10\\
4&8&10&11&
6&7&8&9&
6&7&9&11&
6&7&10&11&
{\bf 6}&{\bf 8}&9&{\bf 12}\\
6&9&11&12&
7&9&11&12&
&&&&
&&&&
&&&\\
\hline
  \end{tabular}}
 \caption{A nucleus with 12 nodes, and 37
tetrahedra, of which 17 have no external face.
If a tetrahedron has an external face, its 3 nodes are shown in boldface.
}\label{table1}
\end{center}
 \end{table}

Our approach is two-fold: Top-down, and bottom-up. In the top-down
approach, we define a set of elementary moves which reduces an
arbitrary triangulation of the 3-ball into a tree of nuclei, which are
glued together by pairs of faces, each such face with 3 external
edges. The tree can then be cut into a disjoint union of nuclei by
cutting along these faces. The construction always transforms 3-balls
to unions of 3-balls, and is thus implementable on a computer.

In the bottom-up approach, we start with any tree whose nodes are
arbitrary nuclei, and we construct 3-balls from it by gluing adequate
faces together. Not all possible gluings lead to 3-balls, but
including also some inadequate gluings still leads to good bounds.
Again, the procedure can be programmed on a computer.

Our main result is \tref{t:main}. It says that \emph{if the number
$\rho(t,f_\s)$ of face-rooted nuclei with $t$ tetrahedra and $f_\s$ external
faces has a bound of the form $\rho(t,f_\s)\le C^{t}$ then
the number of rooted triangulations of the 3-ball with $t$
tetrahedra, $f$ external faces and $n$ internal nodes is bounded by
$C_*^{t+f+n}$.}

In particular, since obviously, $f\le 4t$ and $n\le 4t$, we would get a
bound $C_{**}^{t}$.

In summary, our work bounds the number of triangulations in terms of
the number of nuclei. Thus, we remain with a new, but hopefully
simpler, open question about the problem posed by Gromov, namely 
does the number of face-rooted nuclei with $t$ tetrahedra have an
exponential bound in $t$\,?
 While we do
not have any mathematical statements about this problem, the
methodology of the proof of \tref{t:main} allows for quite extensive
numerical experimentation. The most important insight from this
experimentation is as follows: It seems that if $T$ is a nucleus with
a $k$-complicated knot (or even braid), then the addition of (at most) $k$ cones
and decomposition with our algorithm leads to a tree of
\emph{tetrahedra}.
Note that the trefoil knot is 1-complicated. Furthermore, Goodrick \cite{Goodrick1968} showed 
that the connected sum of $k$ trefoil knots is at least $k$-complicated.

We have analyzed a certain number of classical examples, with the
following findings summarized in Table~\ref{tab:tab2}.
\begin{table}[h]
  \begin{center}
    \scalebox{1.0}{
\begin{tabular}{|l|l|l|c|}
\hline
Example&knot complication&\# of cones added&ref.\\
\hline
Bing 2 room& no knot& 1 cone&\cite{Bing1964}\\
1 trefoil& 1-complicated& 1 cone&\cite{Furch1924}\\
2 trefoils& 2-complicated& 1 cones&\\
3 trefoils & 3-complicated& 2 cones&\cite[Figure 3]{BZ2011}\\
4 trefoils & 4-complicated& 3 cones&\\
5 trefoils & 5-complicated& 3 cones&\\
figure 8 & 1-complicated &1 cones&\\
cinquefoil knot & 1-complicated & 1 cones\\
\hline
\end{tabular}}
  \end{center}
\caption{Experimental upper bound on the number of cones needed to
  decompose a triangulation into 
  tetrahedra (For the definition of $m$-complicated, see \cite{BZ2011}).}\label{tab:tab2}
\end{table}

\subsection{The method}The bounds on the number of triangulations are
obtained by studying a set of elementary moves, detailed in
\sref{s:elementary}. These moves either decompose the triangulation in
two disjoint pieces (by cutting along an interior face with 3 edges on
the boundary, or taking away a tetrahedron with an external face and
one internal node). Clearly, this leaves again two 3-balls on which we
continue the decomposition. The other operations are ``open'' a ball along
a carefully chosen edge (which we call ``split-a-node-along-a-path'')
or opening one face with 2 external edges. These operations
\emph{increase} the number of tetrahedra in the triangulation, but they
prepare the moves in which the 3-ball can be cut, and the internal
nodes can be eliminated. One of the main novelties of this
construction is the observation that this can be done with \emph{few}
additional tetrahedra: This follows from a careful analysis of cuts of
the 2-dimensional hemisphere attached to any external node. Since this
is an important bound, we devote \sref{s:twocolor} to its proof.
In \sref{s:notation}, we introduce the
(standard) terminology for the pieces of any triangulation. In
\sref{s:part1} we combine the 4 moves described above to show how a
general triangulation can be decomposed into a set of nuclei. In
\sref{s:part2}, we perform the bottom-up procedure and show how one
bounds the number of triangulations of the 3-ball in terms of trees
whose nodes are (rooted) nuclei, extending in this way the earlier
work of \cite{DJ1995} and \cite{BZ2011}.

\subsection{Comparison with 2d}
It is useful to compare our method to what can be done in 2d.
In 2d we have a set of triangles. Any triangulation can be obtained in
the following way: First, construct a tree of triangles, adding each
triangle with only one face to the existing tree. This object has no
internal nodes. Now, glue together adjacent faces of the tree,
recursively. In this way one can obtain all triangulations of any polygon.

The inverse operation, while intuitively clear, is a little harder to
describe, and we just sketch the procedure. Given any internal node
$x$ at
distance 1 from the polygon, say connected to $n$ we can split the
edge $(n,x)$ by doubling the node $n$ into a pair $n'$, $n''$, so that
the edges $(n',x)$ and $(x,n'')$ are now external edges and $x$ is
promoted to an external node. All internal nodes can recursively be
brought to the surface in this way. We then have a tree, and the tree
can be decomposed into triangles by cutting all internal edges with 2
external nodes. At the end, the basic objects are triangles.

Clearly, therefore, the basic objects in 2d are
\begin{myenum}
  \item[2a)]internal edges with 2 external nodes
\item[2b)] internal nodes (at distance 1) from the polygonal boundary 
\end{myenum}

In 3d, there are many more possibilities, and our procedure will
eliminate all those which can be eliminated. The ones which we can
deal with are
\begin{myenum}
  \item[3a)] internal faces with 3 external edges: this corresponds to case
    2a) above and will be cut by cut-a-3-face
\item[3b)] internal faces with 2 external edges, and therefore one internal
  edge with 2 external nodes. This resembles 2b) and is dealt with by
  open-a-2-face.
\item[3c)] an internal node $x$ which is the tip of a tetrahedron $t$ whose
  opposite face is external. One can just eliminate $t$
  and $x$ becomes external. This is the second case which corresponds
  to 2b). We call this C0 later.
\item[3d)] an internal node $x$ which is the corner of a face $f$ whose
  opposite edge is external (but not C0). Again, a sub-case of
  2b). This is dealt with split-a-node-along-a-path, and will be
  called C1.
\item[3e)] an internal node $x$ which is the end of an edge $e$ whose
  opposite end is external (but not C1). Again, a sub-case of
  2b). This will be called C2 and reduced to C1 with
  split-a-node-along-a-path.
\end{myenum}

The elementary objects are those left over after all these
decompositions are performed. In 2d, those objects are just triangles, which makes
the counting possible. In 3d these are nuclei. Non-trivial nuclei
exist, and they must carry the information about the
complications of 3 dimensional topology, since all the other problems
have been eliminated. In particular, internal faces of nuclei have 0 or 1 external edges.

\section{General definitions and notations}\label{s:notation}

\subsection{Internal and external objects, flowers}\label{s:objects}

To be precise, we redefine here some terminology which is common in
the discussion of triangulations. We start with triangulations of
$S^2$. These will have $f_\s$ faces, $n_\s$ nodes and $e_\s$ edges,
where the subscript s stands for ``surface''. This triangulation is the
boundary of a ball which is filled with tetrahedra, some of which have
faces among the $f_\s$ external faces. We call this also a triangulation, and
we say that $t$ is the number of tetrahedra, $f_{\rm tot}$ the number of faces,
$e_{\rm tot}$ the number of edges, and $n_{\rm tot}$ the number of nodes.
Faces, edges,
and nodes which are not among those of the triangulation of $S^2$ are
called \emx{internal}. It will be useful to observe that tetrahedra can have
up to 4 external faces, internal faces can
have up to 3 external edges, internal edges up to 2 external nodes.
We will use the subscript i for internal objects.

Obviously,
\begin{equ}
  f_{\rm tot}=f_\s+f_\i~,\qquad
  e_{\rm tot}=e_\s+e_\i~,\qquad
  n_{\rm tot}=n_\s+n_\i~.
\end{equ}
{}From the Euler relations and trivial geometry, we have the relations
\begin{equa}[e:euler]
  t-f_{\rm tot}+e_{\rm tot}-n_{\rm tot} &=-1~,\\
f_\s -e_\s +n_\s &=2~,\\
3f_\s&=2e_\s~,\\
4t&=2(f_{\rm tot}-f_\s)+f_\s~.
\end{equa}
This leaves us with 3 free variables, which we choose as
\begin{equa}
  t, f_\s , \text{and~} n_\i~.
\end{equa}
Note that $f_\s$ is always even.
\begin{definition}
  We use the term \emx{f-vector} for the three variables $\langle
t,f_\s,n_\i\rangle$ where $f_\s \ge 4$.
\end{definition}

\subsection {Notation and flowers}
We introduce some notation which we apply to triangulations and
tetrahedrizations (which we also call triangulations when no confusion
is possible):
\begin{myitem}
 \item If $n_1$ and $n_2$ are 2 distinct nodes, then we denote by
$(n_1,n_2)$ the edge connecting the two.
 \item Similarly, if $n_i:i=1,2,3$ are 3 distinct nodes, then
$(n_1,n_2,n_3)$ is the face (triangle) with these 3 corners.
\item If $e$ is an edge and $n$ is a node not in $e$ then $(n,e)$
  denotes the face (triangle) with the edge $e$ and the node $n$.
\end{myitem}
This notation is easily generalized to the case where we consider simplices
of dimension 3:
\begin{myitem}
 \item If $n$ is a node and $f$ is a face not containing $n$, then $(n,f)$ is
the tetrahedron with $f$ as a face and $n$ as
the opposite corner.
 \item Similarly, if $e$ is an edge and $n_1,n_2 \notin e$ are 2 distinct nodes
then $(n_1,n_2,e)$ is the unique tetrahedron containing all of
them.
 \item Finally, if $e_1$ and $e_2$ are two edges without common nodes,
   then $(e_1,e_2)$
is the tetrahedron containing both edges.
\end{myitem}
Paths of nodes connected by edges will be denoted as $\gamma =
[n_1,n_2,\dots,n_k]$ and the union of 2 disjoint paths $\gamma_1$,
$\gamma_2$ (connected by
one or both
endpoints) will be denoted by  $\gamma_1\cup\gamma_2$.\newline

We next define what we mean by \emx{flowers}. Here, we adapt the conventions
to the tetrahedrization of  a triangulated sphere $S^2$. Nodes,
edges, and faces are called \emx{external} if they lie entirely in
$S^2$.
All others are called \emx{internal}.
Consider an external node $n_*$.\footnote{We use $n_*,m_*$ and the
  like for external nodes, and $x_*,y_*,\dots$ for internal ones.} We
define its 2 flowers:
\begin{myitem}
 \item The \emx{external flower} $\EE(n_*)$ of $n_*$ is the set of all edges
$e$ not containing $ n_*$ for which $(n_*,e)$ is an external face.
Clearly, $\EE(n_*)$ is
a polygon.
 \item The \emx{internal hemisphere} $\II(n_*)$ of $n_*$ is the set of all
faces
$f$ not containing $n_*$ for which $(n_*,f)$ is a tetrahedron. It is
   easy to see that $\II(n_*)$ is a
2d triangulation
whose boundary is the polygon $\EE(n_*)$.
\end{myitem}

We will say that the external flower of an \emph{internal} node $x_*$ is empty.
The internal (hemi-)sphere $\II(x_*)$ (or simply flower) of $x_*$
is a triangulation of $S^2$.

We also define the \emph{external flower} $\EE(e)$ \emph{of an external edge}
$e$ as the 2 nodes $n_1,n_2$ for which $(n_i,e)$ are 2 external faces.
Similarly, the \emph{internal hemisphere} $\II(e)$ \emph{of the external edge}
$e$ is defined
as the set of all edges $e'$ such that $(e,e')$ is a tetrahedron. By
hypothesis,
$\II(e)$ is a 1-d triangulation whose boundary is $\EE(e)$.
Note that there might be internal nodes at distance 1 from $e$ which
are not in $\II(e)$.

\section{Some geometrical considerations: Two-colored paths in a
triangulation}\label{s:twocolor}

We describe here properties of paths in a 2d triangulation of a
polygon. These properties will play a crucial role when we will bound
the effects of moving internal nodes of a 3d triangulation to the
surface. However, they are totally independent of the remainder of the
paper.

\begin{lemma}\label{l:euler2d}
 Let $\KK$ be a 2d triangulation of a $p$-gon $P$ with $n$ interior
 nodes. Then the
number of interior edges in $\KK$ is $3n+p-3$.
\end{lemma}
\begin{proof}
  The proof follows from the Euler relations and is left to the reader.
\end{proof}
\begin{lemma}\label{l:3connected}
 Consider a polygon $P$ and let $\KK$ be any triangulation of $P$ with $k>0$
internal nodes. For each node $x \in \KK \setminus P$, there are at least
 3 simple disjoint paths in the interior of $\KK$ connecting it to 3
 different points of $P$.
\end{lemma}
\begin{proof}
 Any triangulation of $S^2$ is 3-connected. Complete $\KK$ into a
triangulation of $S^2$ by adding a cone over its boundary. Let $m$ be the apex
 of the cone. Then there are at least 3 disjoint simple paths connecting $x$ to
$m$, \cite{Diestel2010}. Any such path must intersect $P$, and we take the
 first intersection point.
\end{proof}

We assume now that the
nodes of $P$ are labeled.

\begin{definition}\label{def:admissible}
  A triangulation $\KK$ is called \emx{admissible} if the following
  conditions are met:
  \begin{myitem}
    \item[K1:]The boundary $\partial\KK$ has at least 2 different
      labels.
      \item[K2:]The nodes with a given label form one connected arc of
        $\partial\KK$ .
      \item[K3:]The ends of any edge connecting 2 nodes of
        $\partial\KK$ have different labels, unless the edge is in
$\partial\KK$~.
  \end{myitem}
\end{definition}

The \fref{fig:konditions} illustrates the definition.
\begin{figure}[h]
  \includegraphics[width=\textwidth]{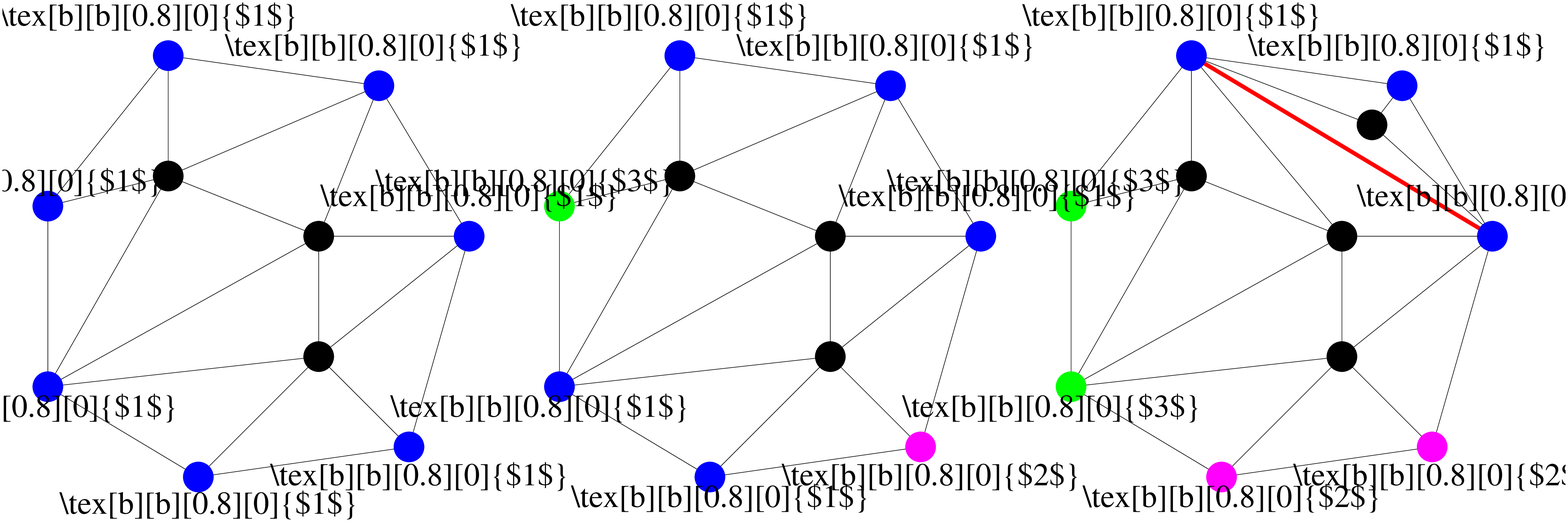}
  \caption{An illustration of the conditions K1)--K3). Left: since
    there is only one label, K1) is violated. Center: The region with
    label 1 is not connected; K2) is violated. Right: There is an
    internal link (red) connecting two nodes with the same label; K3)
    is violated.}
\label{fig:konditions}
\end{figure}

We first need an auxiliary lemma.
We will need the following information:
\begin{lemma}\label{l:crisscross}
Let $\KK$ be as above and let $P=\partial\KK$.
Given two boundary nodes $a$ and $b$ with
  different labels at least one of the two alternatives below holds:
  \begin{myitem}
\item[1)]There is a simple path $\gamma$
  joining $a$ and $b$ without any other node in $P$,
\item[2)]
There is an  edge $(x,y)$ joining the two pieces of $P\setminus\{a,b\}$.
  \end{myitem}
\end{lemma}
  \begin{figure}
\begin{center}
     \includegraphics[width=0.8\textwidth]{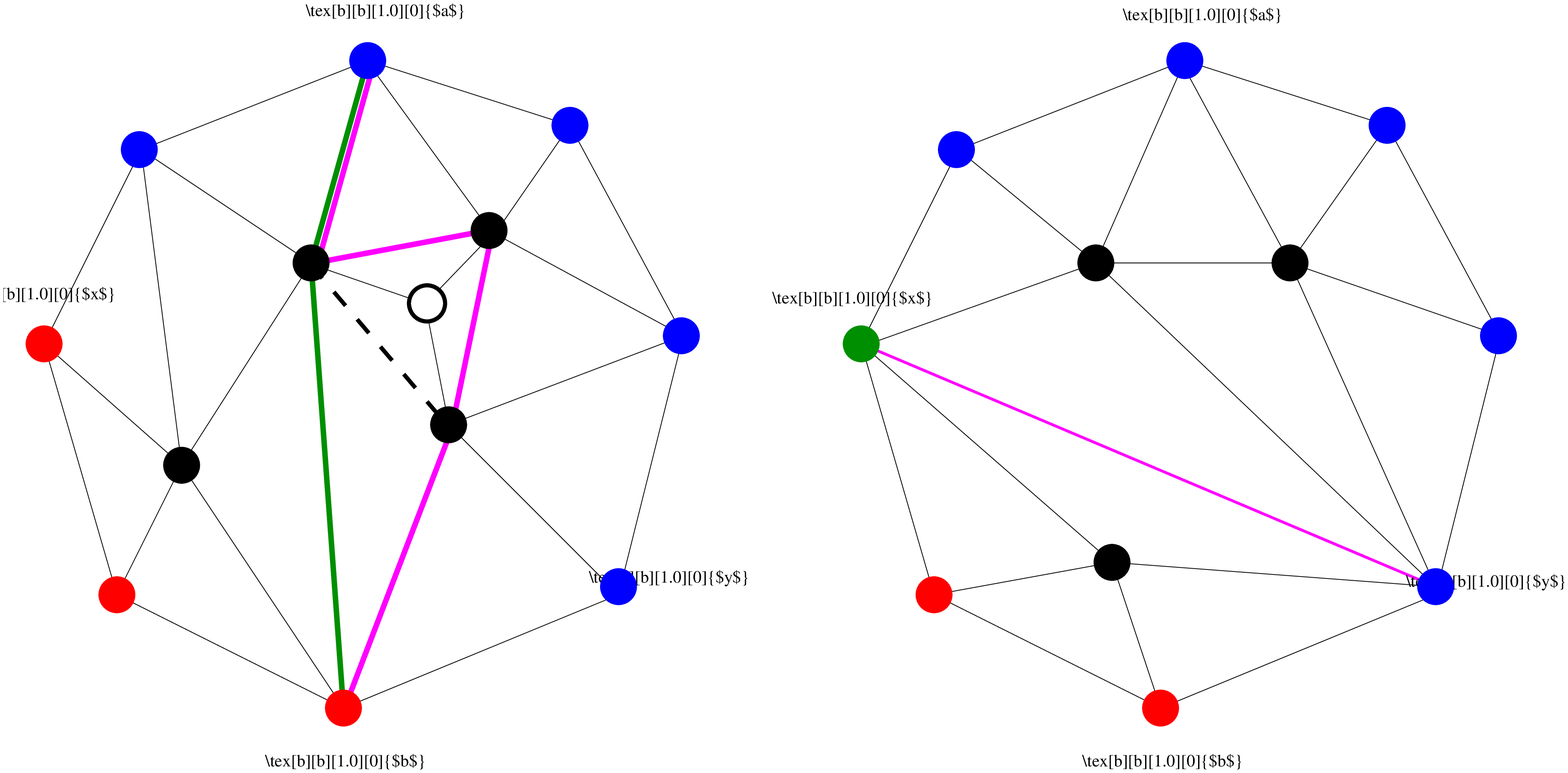}
\end{center}
    \caption{The 2 alternatives of finding a path connecting two
      different labels. Left: There is an interior path between $a$
      and $b$. Right: There is no such path, but then, one can always
      find an edge connecting two different labels (by K3), (not necessarily
      the same as $a$ and $b$). The left panel also illustrates the
      necessity of choosing a shortest path. For example, choosing the
      magenta path, the dashed edge will violate K3) in the next step
      of the procedure.}\label{fig:crisscross}

  \end{figure}

Postponing the proof of \lref{l:crisscross} we have

\begin{proposition}\label{prop:twocolor}
  Assume $\KK$ is an admissible triangulation in the sense of
  \dref{def:admissible} with at least 2 triangles.
Then, there exists a path $\gamma$ along
  internal edges of $\KK$ which connects two points in $P=\partial\KK$
  with \emph{different labels}. It cuts $\KK$ in two
  pieces $\KK_\L$ and $\KK_\R$. The path $\gamma$ can be chosen in
  such a way that labeling the new boundary piece (namely the
  interior nodes of $\gamma$) in
  $\KK_\L$ and $\KK_\R$ with a label different from the ones used so far, both
  $\KK_\L$ and $\KK_\R$ are
again admissible.
\end{proposition}

\begin{proof}
    Let $P=\partial\KK$. By admissibility, we know that not all nodes on $P$
  have the same label. Take nodes $a$ and $b$ with different labels
  and apply \lref{l:crisscross}. If 2) holds, then we take $\gamma $
  as the edge
  connecting $x$ and $y$. By K3) they have different
  labels. Otherwise, there is an interior path connecting $a$ and
  $b$. We take a shortest path, $\gamma$.

Cutting along the path $\gamma$, we obtain 2 pieces $\KK_\L$ and
$\KK_\R$. If $\gamma$ is just one edge then inspection shows that
K1)--K3) hold. In the second case, K1) and K2) are obviously true by
construction. Giving a new label, say $L$, to the interior nodes of
$\gamma$, we have to show that there are no edges connecting any two
non-consecutive nodes with label $L$. But if there were, the path
would not be minimal.
\end{proof}

\begin{proof}[Proof of \lref{l:crisscross}] The reader may want to look at
\fref{fig:crisscross}.
  Assume 1) does not hold.
This means that one cannot draw 3 disjoint paths between $a$ and $b$, as
the middle one would satisfy 1). We can take the two disjoint paths to go
along the two boundary segments between $a$ and $b$. By Menger's
theorem
\cite{Diestel2010} there must then be 2 nodes $x$ and $y$ (other than $a$ or $b$) such
that all paths
from
$a$ to $b$ must pass through at least one of them.
Since the boundary paths are candidates, we see that $x$ and $y$
are in $P$, one per arc connecting $a$ and $b$.
Consider now the path from $a$ to $b$ along $P$ which goes through $x$.
Modify it so that instead of going through $x$ it goes through the flower of
$x$. We get a new path from $a$ to $b$ which does not go through $x$. This
means that the new path goes through $y$ implying that $y$ is in the flower of
$x$. Thus, $x$ and $y$ are connected by an edge.

This completes the proof.
\end{proof}

\section{Part I: Reducing any triangulation into a set of
nuclei}\label{s:part1}

\subsection{The elementary moves}\label{s:elementary}

In this section we define the elementary moves which
transform any
triangulation into a (set of) nuclei.
The first two moves, which we call \emx{open-a-2-face} and \emx{cut-a-3-face},
are used to transform any triangulation with no internal nodes into
a set of nuclei, and
the third and fourth move, which we call \emx{remove-1-tetra} and
\emx{split-a-node-along-a-path}, are used
to remove all internal nodes of a triangulation.

Henceforth, $T$ will denote a triangulated 3-ball with f-vector
$\langle t,f_\s,n_\i\rangle$.

\subsubsection{Cut-a-3-face}
Let $(n_1,n_2,n_3)$ be an internal face with its 3 edges on the
surface $\partial T$ of $T$. Then, it cuts the 3-ball $T$
into 2 distinct parts. We simply separate these 2 parts and we get 2
``smaller'' 3-balls.
In other words, we know that any triangulation is (uniquely) defined by the
list of all its tetrahedra.
We find the two lists corresponding to the tetrahedra which are on either side
of
the
internal face in question, and we define 2 new 3-balls, each
associated with one
of these 2 lists.
If $\langle t,f,n\rangle $, $\langle t_1,f_1,n_1\rangle $ and $\langle
t_2,f_2,n_2\rangle $ are the f-vectors of the initial ball and the 2 new ones,
then we have
\begin{equs}
 t = t_1+t_2~,\qquad
 f = f_1+f_2-2~, \qquad
 n = n_1+n_2~.
\end{equs}

\subsubsection{Open-a-2-face}

Consider 3 external nodes $n_*,n_1,n_2$ of $T$ which form a triangle
$(n_*,n_1,n_2)$. We assume that $(n_*,n_1,n_2)$ is an internal
face, with $(n_1,n_2)$ an internal edge, and the two other edges external.
Let $\II$ and $\EE$ be the internal and external flower
of the external node $n_*$. As we have already stated, $\II$ is a triangulation
of the polygon $\EE$. By hypothesis, the edge $(n_1,n_2)$ divides $\II$ into 2
distinct sets of faces.
The operation open-a-2-face consists in removing $n_*$ and all tetrahedra
attached to it, replacing it by $n_{*,1}$ and $n_{*,2}$ and attaching
each of these 2 new nodes to all faces of one of the two parts of $\II$, see
\fref{fig:open2}.
\begin{figure}[t]
\begin{center}
   \includegraphics[height=12cm]{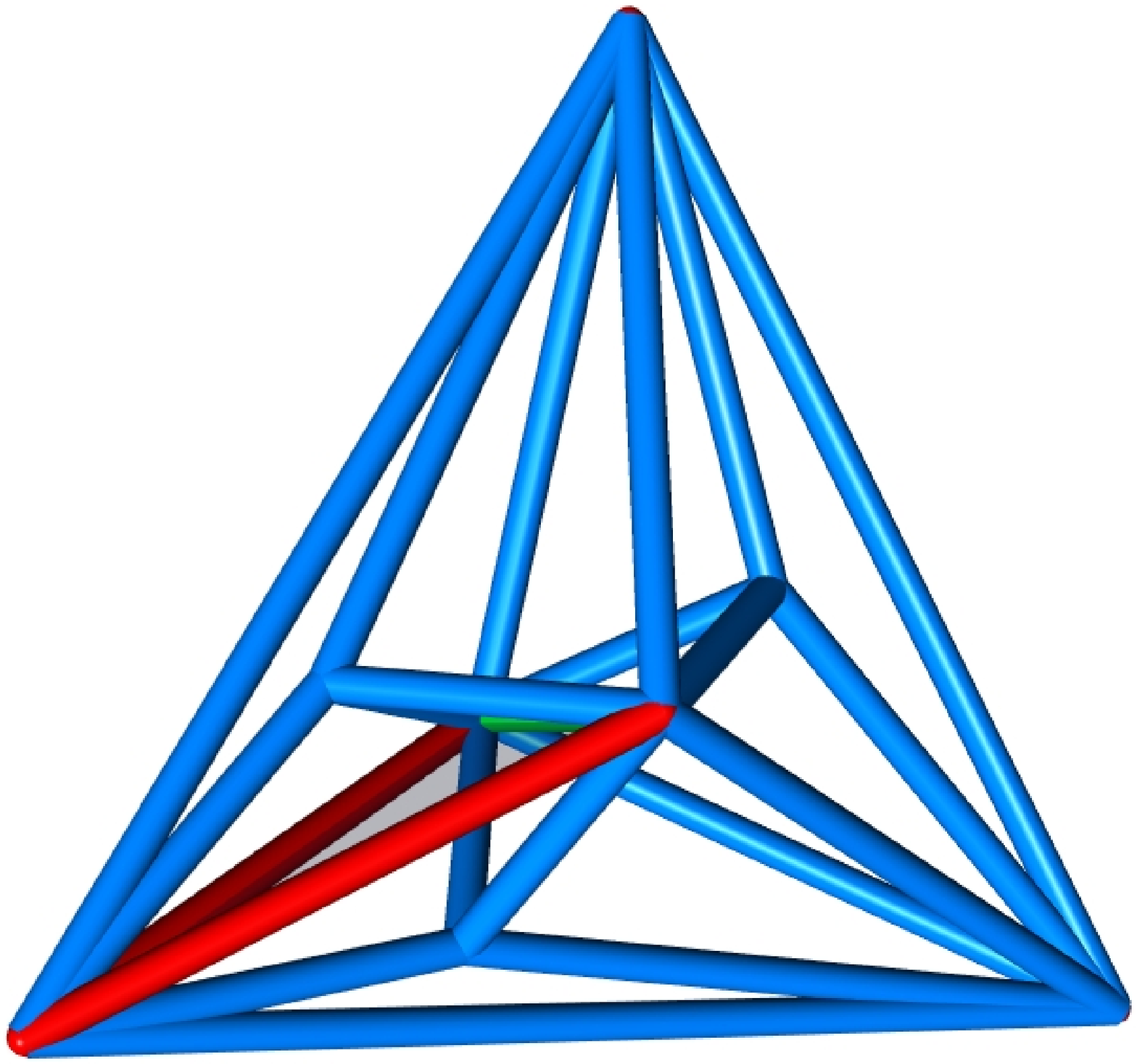}
   \includegraphics[height=12cm]{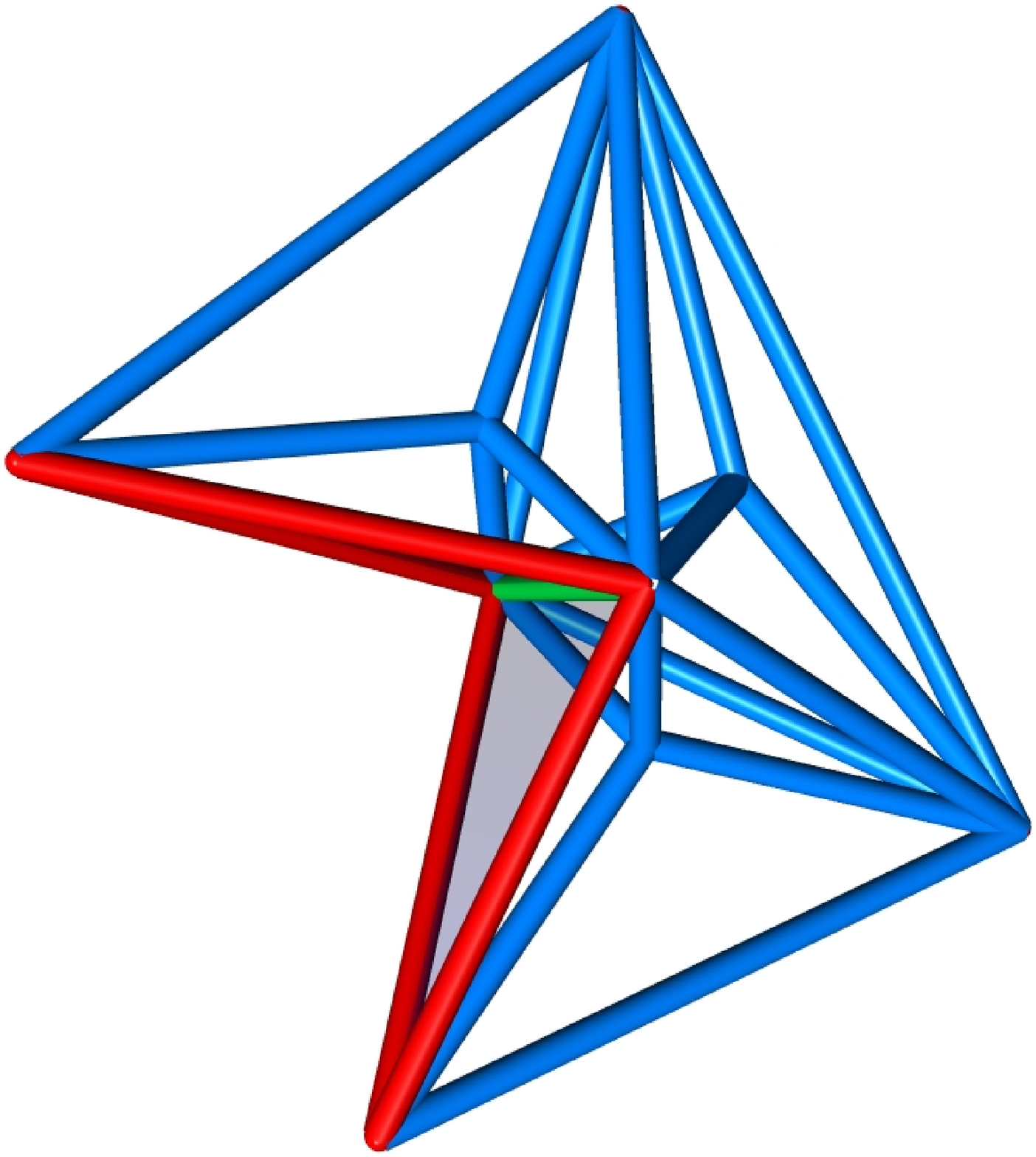}
\end{center}
  \caption{Sketch of open-a-2-face.}\label{fig:open2}
\end{figure}
\emph{This operation transforms a triangulation of the 3-ball into a
  triangulation of the 3-ball.}
If $\langle t,f_\s,n_\i\rangle $ and $\langle t',f_\s',n_\i'\rangle $ are the
f-vectors of
the initial ball and the resulting ball, then we have
\begin{equ}
 t' = t~,\qquad
 f'_\s = f_\s+2~,\qquad
 n'_\i = n_\i~.
\end{equ}
We will say that the f-vector changes by $\langle 0,+2,0\rangle$.

\subsubsection{Remove-1-tetra}\label{s:r1t}
\begin{definition}
 A \emx{removable tetrahedron} is any tetrahedron $t$ with one
 internal node and one external face.
\end{definition}

The operation remove-1-tetra is as
follows:
let $t_*=(x_*,n_1,n_2,n_3)$ be a removable tetrahedron with internal node
$x_*$ and external face $(n_1,n_2,n_3)$.
We simply remove $t_*$ and its external face; the internal node $x_*$, the 3
internal edges and the
3 internal faces of $t_*$ all become external.
The f-vector $\langle t,f_\s,n_\i \rangle$ changes to $\langle
t-1,f_\s+2,n_\i-1
\rangle$; the change of f-vector is $\langle -1,2,-1\rangle$.

\subsubsection {Split-a-node-along-a-path, hemispheres and
pieces}\label{s:def-split}

Consider an external node $n_*$ of $T$ and its internal hemisphere
$\II=\II(n_*)$, see \fref{fig:after-split} for an illustration.
By definition of a triangulation, $\II$ is a 2d triangulation of a polygon.
\begin{definition}
 A \emx{splitting path} $\gamma$ is any simple path which connects two
different points on
 $\partial\II$ and contains no edge of $\partial \II$.
\end{definition}

Let $\gamma$ be a splitting path. Clearly it divides $\II$ into 2 pieces
$\KK_\L$ and
$\KK_\R$ with $\II=\KK_\L\cup\KK_\R$ and
$\KK_\L\cap\KK_\R=\gamma$.

The move
split-a-node-along-a-path $\gamma$ is defined as follows:
\begin{myenum}
  \item Remove the node $n_*$ and all tetrahedra having $n_*$ as a corner
  \item Add 2 new nodes $n_{*,\L}$ and $n_{*,\R}$
  \item For each face $f_*\in \KK_\L$ add the tetrahedron $(n_{*,\L},f_*)$
  \item For each face $f_*\in \KK_\R$ add the tetrahedron $(n_{*,\R},f_*)$
  \item For each edge  $e\in \gamma$ add the tetrahedron
$(n_{*,\L},n_{*,\R},e)$
\end{myenum}

Note that by construction,
 one of the nodes on $\partial \KK_\L$
    is $n_{*,\R}$, and the links in $\KK_\L$ starting from $n_{*,\R}$ reach
    (the image of) $\gamma$.  Analogous statements hold for
    $\KK_\R$.

    \begin{definition}
      In the construction above, we refer to $\KK(n_{*,\L})=\KK_\L$ as the left
\emx{piece}  .
      It is simply the subgraph obtained from the hemisphere $\II(n_{*,\L})$
after removing the cone connecting $n_{*,\R}$ to every node of $\gamma$.
      Similarly, we define the right piece $\KK_\R$.
    \end{definition}

  \begin{remark}
   Hemispheres $\II$ and pieces $\KK$ will play an important role in our
construction. Some statements will be given for hemispheres, others for pieces
and so
   it is important to be able to distinguish between the two definitions.
  \end{remark}
 \begin{remark}
   A splitting path $\gamma$ is always associated with a hemisphere $\II$ and
not with a piece $\KK$. We will see that, under some conditions,
   a simple path $\tilde \gamma$ connecting
   two nodes of the boundary of a piece $\KK$ can be extended into a splitting
path $\gamma$.
  \end{remark}

  \begin{figure}[h]
    \begin{center}
      \includegraphics[width=\textwidth]{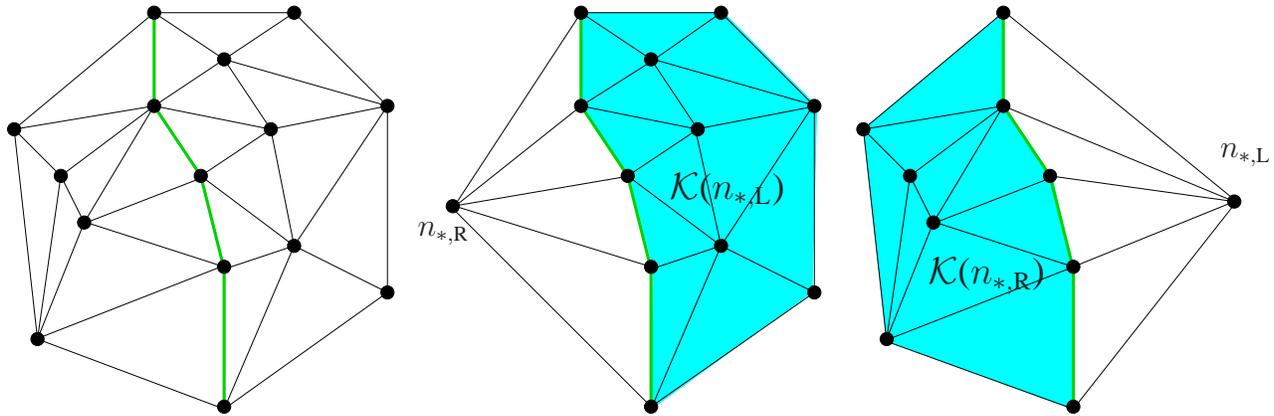}
      \caption{The left panel shows the internal hemisphere $\II(n_*)$ of
$n_*$. We
split $n_*$ into $n_{*,\L}$ and $n_{*,\R}$ along the green path $\gamma$. The
        other 2 panels show the internal hemispheres $\II(n_{*,\L})$ and
$\II(n_{*,\R})$ of the 2 new nodes. Notice that
        each internal node of the green path $\gamma$ is at distance 1 from
        $n_{*,\R}$ resp.~$n_{*,\L}$
        Also, the links leaving $n_{*,s}$, $s\in\{\L,\R\}$ have been added
        during the split.
}
      \label{fig:after-split}
    \end{center}
  \end{figure}

\begin{lemma}\label{l:split}
  The move split-a-node-along-a-path transforms a 3-ball into a
  3-ball. The f-vector $\langle t,f_\s,n_\i \rangle$ is mapped
  to $\langle t+|\gamma|,f_\s+2,n_\i \rangle$, where $|\gamma|$ is
  the number of
  edges in $\gamma$.
\end{lemma}
The f-vector changes by $\langle |\gamma|,+2,0 \rangle$. In
particular, the number of tetrahedra \emph{increases}. But we will
show that this increase can be controlled.

\begin{proof}

The count of the f-vector
is as follows: Removing and adding the tetrahedra in steps 1,3,4 above
does not change their number. The number of external faces increases
by two, namely the two external faces sharing the new edge
$(n_{*,\L},n_{*,\R})$. And each internal face $(n_*,e)$ which connected $n_*$
to an edge $e$ in
$\gamma$ gives rise to a new tetrahedron $(n_{*,\R},n_{*,\L},e)$.
 There are $|\gamma|$ such faces and so the f-vector is seen
to change by $\langle |\gamma|,+2,0 \rangle$, as asserted.

\end{proof}

\subsection{Summary}

In the sequel, we want to bound the effect of removing internal nodes,
since our building blocks are the nuclei, which do not have any
internal nodes. Eliminating the internal nodes will cost the addition
of tetrahedra, and the issue here is how many are needed to obtain a
ball without internal nodes. Internal nodes disappear when we perform
the remove-1-tetra operation, and only then.

Before starting the bounds proper, we explain here the point of our
construction, based on the evolution of the f-vectors $\langle
t,f_\s,n_\i\rangle$. Open-a-2-face costs a change $\langle
0,2,0\rangle$, and split-a-node-along-a-path costs
$\langle |\gamma|,2,0\rangle$, where $|\gamma|$ is the length of the path
along which we cut. 
In principle, each path $\gamma$ might have a length proportional to
the number of nodes, which in turn would imply that the sum of the
lengths of all paths exceeds $\OO(n_{\rm tot}^2)$.
So one needs a strategy which improves this naive bound.

While we cut, new
external edges appear, and also, new external edges appear when we
remove a tetrahedron which costs $\langle -1,2,-1\rangle$. But it is only
this operation which reduces the number of internal nodes. So, there
are two opposing tendencies. One is the preparation of promoting an internal node into 
an external one, and it \emph{adds} many tetrahedra, and the
other is remove-1-tetra, which reduces the number of internal nodes by
1. 
The real issue is thus to bound the \emph{number of added tetrahedra
per removed 
internal node}. We will perform this bound in terms of the number $e_\s$
of internal edges.
Our main result is Corollary \ref{c:increase} which says that the number of
internal edges grows by no more than $C_\Delta(t+n_\i)$.
The Euler relations \eref{e:euler} allow to express $t$ as a function
of $e_\s$, $f_\s$, and $n_\i$,
\begin{equ}
  t=e_\i-n_\i+f_\s/2-1~.
\end{equ}
Therefore, and since $n_\i <4t$ and $f_\s<4t$, Corollary~\ref{c:increase} implies that the elimination
of all  $n_\i$ internal
nodes leads to an f-vector of the form
\begin{equ}
 \langle t,f_\s,n_\i\rangle  \to \langle t' , f_\s',0 \rangle~,\qquad 
f_\s'< C\cdot t~,\quad  t' < C\cdot t~,
\end{equ}
with a finite constant $C$ which is \emph{independent} of the triangulation.

\subsection{Removing internal nodes}\label{s:intnodes}

 \subsubsection{Definitions and strategy}\label{s:strategy}

Given any triangulation, we define the \emx{depth} $D_x$ of a node $x$ as the
minimal number of connected edges needed to reach the boundary, starting from
$x$. The strategy
will consist in recursively reducing the depth of any internal node by
1. This is repeated until no internal nodes remain.

We classify an internal node $x_*$ at depth 1 in 3 flavors, which we
call C0--2:
\begin{myenum}
 \item[C0:] $x_*$ is the internal node of a removable tetrahedron.
 \item[C1:] $x_*$ is not of type C0 but is in a face $(x_*,n_*,m_*)$
   where $(n_*,m_*)$ is an
   external edge.
 \item[C2:] $x_*$ is neither of type C0 nor C1.
\end{myenum}

We enumerate the external nodes in an arbitrary order, leading to a
list $\LL _{0}=\{n_{*,1},\dots,n_{*,k}\}$.
Similarly, for $d>0$, we define $\LL_d$ as the nodes at depth $d$
from the surface.
Given an
external node $n_*\in\LL_0$, we consider its hemisphere
$\II(n_*)$.

An internal node $x_*\in\II(n_*)$ of type C2 can (only) be
\emx{promoted} to an internal node
of type C1 by drawing a path $\gamma \subset \II (n_*)$ that goes through it
and
splitting $n_*$ into $n_{*,\L}$ and $n_{*,\R}$ along $\gamma$.
Indeed, one easily sees that $n_{*,\R} \in \EE(n_{*,\L})$ and
that $(n_{*,\L},n_{*,\R},x_*)$ is a face.

In the same manner, we see that a node $x_*\in\II(n_*)$ of type C1
can be promoted into an internal node of type C0 by drawing a path
 $\gamma \subset \II (n_*)$ which contains the edge
$(x_* , y)$.  Here, $y$ is the external node of the face $(x_*,n_*,y)$ which
defines $x_*$ as a node of type C1. Splitting $n_*$ along $\gamma$,
the tetrahedron $(n_{*,\L},n_{*,\R},y,x_*)$ becomes removable.

Finally, any internal node of type C0 can be made external by simply
removing one tetrahedron.

The strategy is in 4 steps (3 sweeps). We set $\LL=\LL_0$.
\begin{myitem}
 \item {\bf Step 1 (Sweep C2$\to$C1)} :  We promote all the $x_*$ of type C2 in
the
   following order:
For each $n_*\in\LL$, we
promote all internal nodes of $\II(n_{*})$ of type C2
 into internal nodes of type C1. We will show that this can be done in
 such a way that
every internal edge of the triangulation $\II(n_{*})$ belongs to at
most 1 of the splitting paths (as defined in \sref{s:def-split}).

When this first step is complete, all internal nodes at depth 1 are of
type C1 or C0. There appears a new set $\MM$ of external nodes
containing the nodes of $\LL$ which were not split and new external nodes
obtained by the splitting.
\item {\bf Step 2 (Sweep C1$\to$C0)} : We promote all the $x_*$ of type C1
  in the following order: For each $n_*\in\MM$, we
promote all promotable internal nodes of $\II(n_{*})$ of type C1
 into internal nodes of type C0\footnote{A node $x_*$ of type C1 can be
   promoted to C0 only if it is connected to a node of
   $\EE(n_*)$. This might not be true for all $n_*$ for which
   $x_*\in\II(n_*)$ but there are at least two  $n_*$ for which $x_*$
   is promotable.}. We will show that this can be done in
 such a way that
every internal edge of the triangulation $\II(n_{*})$ belongs to at
most 1 of the splitting paths (as defined in \sref{s:def-split}).

 \item {\bf Step 3 (Sweep C0$\to$external)} : Finally, we make each node of type C0
external by
   removing one tetrahedron.

\item {\bf Step 4} : At this point every internal node has been moved
  up one level of depth. In particular, we let $\LL$ denote those
  nodes which have moved to the surface in step 3. (If the current step is at level $d$, then this set 
equals $\LL_{d+1}$.)

We continue until no internal nodes are left.
\end{myitem}

Since the depth of any node is bounded,
the procedure will end after a finite number of recursive steps.

\subsubsection{Reducing C2-nodes to C1-nodes}\label{s:c21}

Given an external node $n_*\in\LL$, we now describe in detail the
recursive algorithm which promotes the
internal nodes of type C2 in $\II=\II(n_*)$ to type C1. This is achieved
by a succession of carefully chosen moves of type
split-a-node-along-a-path.

  \begin{figure}[h]
    \begin{center}
      \includegraphics[width=0.9\textwidth]{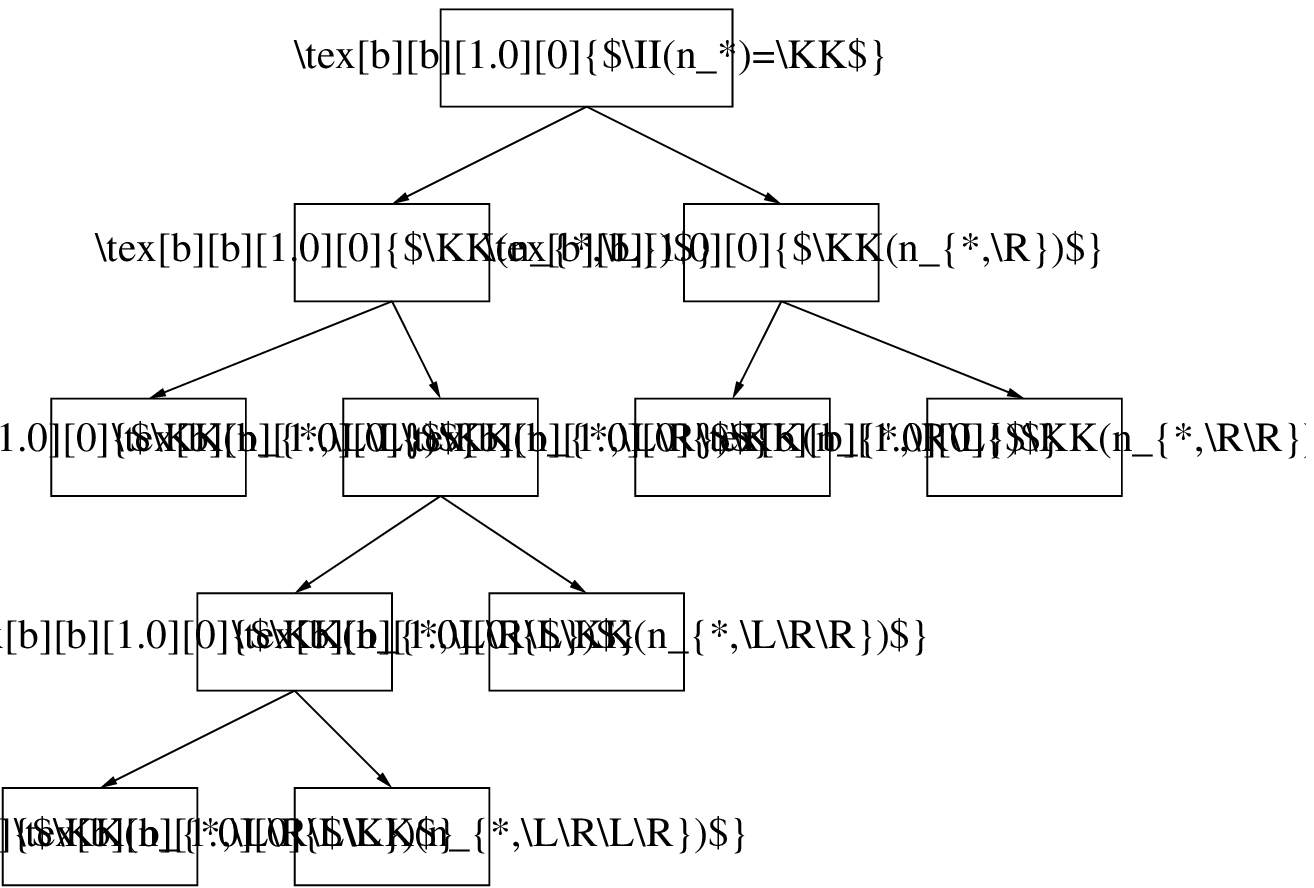}
      \caption{An example of a binary tree of pieces associated with
        the hemisphere of an external node $n_*$ containing 6 triangles.}
      \label{fig:binarytree}
    \end{center}
  \end{figure}

  \begin{figure}[h]
    \begin{center}
      \includegraphics[width=\textwidth]{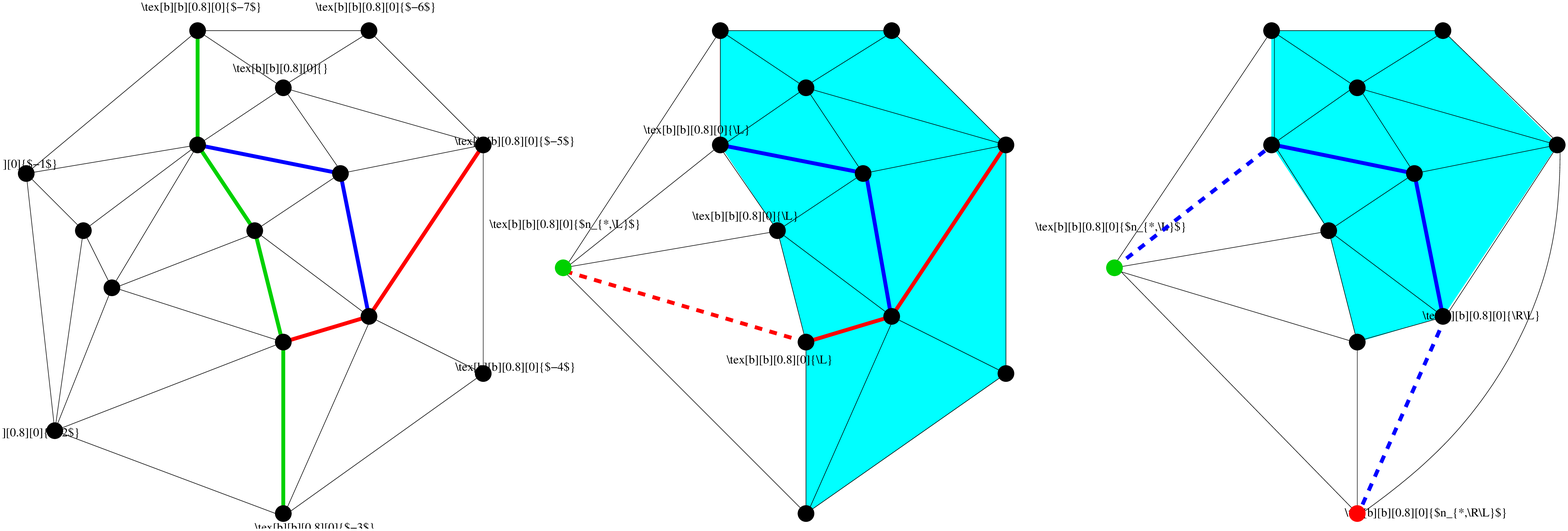}
      \caption{The left panel shows the internal flower $\II(n_*)$ of $n_*$. We
split it in succession along the green, red and blue paths.
        We first split $n_*$ into $n_{*,\R}$ and $n_{*,\L}$ along the green
path. The
middle panel shows $\II(n_{*,\R})$ and the green node
        is $n_{*,\L}$. The shaded region is $\KK_{\R}$ and the new
        labels are $\L$. One end of the red path has a negative label,
        while the other has the label $\L$ and must therefore be
        connected to $n_{*,\L}$.
We obtain the cutting path $\gamma_\R$, and after the cut, we
        obtain two pieces $\KK_{\R\R}$ and $\KK_{\R\L}$. In the third
        panel we show the hemisphere of $n_{*,\R\R}$. The blue path
        has labels $\L$ and $\R\L$ at its extremities, which must
        therefore be connected to $n_{*,\L}$ and $n_{*,\R\L}$. This
        defines the cutting path $\gamma_{\R\R}$. Note that it always
        suffices to add at most 2 dashed segments.
}
      \label{fig:newpaths}
    \end{center}
  \end{figure}

Each of these cuts produces a ``left'' and a ``right'' piece, which
are then cut again into left and right pieces, until only triangles
remain. The pieces will be noted by $\KK_{\S}=\KK(n_{*,\S})$, where $\S$ is a
sequence of letters $\L$ and $\R$ which designate the successive
choices of left and right.

Thus, we construct a binary tree of pieces (see \fref{fig:binarytree}).
In detail:

\begin{myenum}
\item Label the nodes of $\partial\II$ from $-1$ to $-|\partial\II|$.
\item
The hemisphere $\II$ is an admissible triangulation in the sense of \dref{def:admissible}.
\pref{prop:twocolor} implies the existence
  of a shortest path $\gamma$ which connects two nodes of
  $\partial\II$ (with different labels). We choose this path $\gamma$.
\item After splitting along this path, $\II$ is divided in two pieces,
  as shown in \fref{fig:after-split}. The two pieces are called $\KK_\L$
  and $\KK_\R$. 
The splitting has replaced  $n_*$ by $n_{*,\L} $ and
  $n_{*,\R}$ and $\II(n_{*,\L})$ is actually just $\KK_\L$ with
  the cone between $n_{*,\R}$ and $\gamma$ added. This also means that
$n_{*,\R}$ is in the external flower $\EE(n_{*,\L})$ of $n_{*,\L}$.
Analogous
  terminology is used for the other half. At this point, $\S$ is equal
  to $\L$ or $\R$, and we continue with $\S=\L$ (and do later $\S=\R$).
\item If $\KK_\S$ is a triangle, we are done (for this branch of the tree).
\item Label all nodes on $\partial \KK_\S$ which had no label with the
  label $\hat\S$, where $\hat\S$ is obtained from $\S$ by exchanging
  the last letter, cf.~\fref{fig:newpaths}. In this way, the newly
  labeled nodes are connected to $n_{*,\hat \S}$ in $\II(n_{*,\S})$.

\item Considering $\KK_\S$, \pref{prop:twocolor} implies the existence
  of a new shortest path $\tilde\gamma_\S$ which connects two nodes of
  $\partial\KK_\S$ with different labels.
\item We extend the path $\tilde\gamma_\S$ as follows: If the end of
  $\gamma_\S$ has a
  negative label, we do nothing, and if the label is some sequence
  $\S'$ we
  connect the end to $n_{*,\S'}$ by one edge. Doing this for both ends
  we obtain a path $\gamma_\S$.
\item Perform a
  split-a-node-along-a-path on $\gamma_\S$ and continue with step
  4 for the pieces $\KK_{\S\L}$ and $\KK_{\S\R}$.
\end{myenum}

\begin{remark}
 The boundary of a hemisphere $\II_\S = \II(n_{*,\S})$ is composed of 2 types
of nodes:
  \begin{myenum}
    \item Nodes with a negative label which are part of the original boundary
$\EE$.
    \item The children $n_{*,\S'}$ of the original node $n_*$, where
      $\S'$ is a sequence of R's and L's.
  \end{myenum}

 The boundary of a piece $\KK _\S$, which is a sub-triangulation of $\II_\S$,
is also composed of 2 types of nodes:
  \begin{myenum}
    \item Nodes with a negative label which are part of the original boundary
$\EE$, and therefore they are part of the boundary of the hemisphere $\partial
\II_\S$ as well.
    \item The other nodes whose label is some sequence $S'$. These nodes satisfy the following two
conditions:
          \begin{myenum}
            \item All nodes of $\partial \KK_\S$ with the same label
form a connected arc of $\partial \KK_\S$.
            \item If a node $y\in\partial \KK_\S$ has the label $\S'$, then $y$
is an internal node of the hemisphere $\II_\S$ seen as a 2d triangulation.
Furthermore,
              $n_{*,\S'} \in \partial \II_\S$ and $(y,n_{*,\S'})$ is an
internal edge of the triangulation $\II_\S$.
          \end{myenum}
  \end{myenum}

\end{remark}

\begin{theorem}\label{t:c2c1}
The algorithm decomposes the triangulation $\II$, and the
  sub-pieces $\KK_\L$, $\KK_\R$ by sequences of paths
  $\gamma_{s_1,\dots,s_k}$  until all the pieces
  $\KK_{s_1,s_2,\dots,s_k}$ with $s_i\in\{\R,\L\}$ are reduced to
  simple triangles. Furthermore,

every edge in $\II\setminus \partial\II$ is in at most
    one path.
\end{theorem}

\begin{proof}
We need to check that the different steps of the algorithm can be
performed. The steps 1--3 follow from the definition of
split-a-node-along-a-path.
Steps 4 and 5 need no verification. Step 6 relies on
\pref{prop:twocolor}, which implies the existence of a (shortest) path $\tilde \gamma _\S$, cutting 
the admissible piece $\KK_\S$ into two admissible pieces $\KK_{\S\L}$ and $\KK_{\S\R}$.

In step 7, we need to make sure that the path $\gamma_\S$ connects
two \emph{different} nodes of $\EE(n_{*,\S})$ which is also
$\partial\II(n_{*,\S})$, to be distinguished from
$\partial\KK(n_{*,\S})$.
The whole construction of labels has been done with this aim in
mind. Note that if a node $u$ has a negative label, we do nothing because
any node $u$ with a negative label is part of the original boundary $\EE(n_*)$,
implying that
if $u\in\II(n_{*,\S})$, then $u\in\EE(n_{*,\S})$ for any child
$n_{*,\S}$ of $n_*$. On the other hand, if
the label is the sequence $\S'$, then by construction (step 5),
 $u$ is connected to $n_{*,\S'}$ with one edge.
Since the labels are different by construction, the path
$\gamma$ is a splitting path, and therefore a cut along it is possible.
In step 8, we need to verify that the cut can
indeed be done, and that the algorithm can be applied to the children
of the $\KK$ which was just cut. But this is the content of
\pref{prop:twocolor}, which shows that the cut can be done in such a
way that the children are admissible in the sense of
\dref{def:admissible}.

Since new paths are always constructed in the interior of $\KK$, and
the $\KK$'s are cut along them, it is obvious that no edge is covered
by more than one path.

\end{proof}

\subsubsection{Reducing C1-nodes to C0-nodes}\label{s:c10}

Let $T$ be a triangulation of a ball.
Consider an external node $n_*$ of $T$ and let $\II = \II(n_*)$ be its
internal hemisphere. Furthermore, assume that all nodes of $\II$ are either
external (with regard to $T$)
or internal of type C0 or C1 but not C2. We will describe an algorithm which
promotes all the internal nodes of type C1 of $\II$ to internal nodes of type
C0.
The approach is somewhat different from that of the previous section. Indeed
promoting an internal node $x$ of type C2 to an internal node of type C1 is
done by splitting some external node $n_*$
along a path going through $x$. However, let $x\in \II(n_*)$ be an
internal node of type C1
and let $(x,y,n_*)$ be an internal face which defines $x$ as C1; by
hypothesis, $y\in\partial\II$. Promoting $x$ to an internal node of type C0
is done
by splitting $n_*$ along a path which contains the edge $(y,x)$.

For every internal node $x$ of type C1 in $\II(n_*)$ we choose one of
the $y\in\partial\II$ for which $(x,y,n_*)$ is an internal face and
call it $y(x)$.
We define
\begin{equ}
  \YY=\{ (x,y(x))~|~x \text{~is C1~} \}~.
\end{equ}
We will eliminate elements in the list $\YY$ by iterating an algorithm
similar to the one in the previous section, until none are left. A
binary tree of left and right pieces will be formed in the process (see
\fref{fig:binarytree}).

At the first step of this algorithm, this tree only contains
one element,
namely the hemisphere $\II$. We will form a tree of $\KK$'s as before,
starting at $\KK=\II$.

The algorithm starts with steps 1 and 2 below, and then repeats
the other steps until it stops.
\begin{myenum}
\item Pick an edge $(x,y)=(x,y(x))\in\YY$.
\item By hypothesis, $y\in\partial\II(n_*)$. By \lref{l:3connected}
  there is a
second, disjoint, simple path connecting $x$ to a node $z\in\partial\II(n_*)$,
$z\neq y$.
This defines a splitting path $\gamma$ connecting 2 distinct nodes $y$ and $z$
of
$\partial\II(n_*)$.
Similarly to the previous section, we split $n_*$ along $\gamma$ into
$n_{*,\R}$ and $n_{*,\L}$. We add the 2 new pieces $\KK(n_{*,\R})$ and
$\KK(n_{*,\L})$
as two leaves of $\KK$ in the tree. We remove the edge $(x,y)$ from the
list $\YY$. Note that the path $\gamma$ might promote a second internal node
$x'$
of type C1 into a node of type C0, if the edge $(x',z)$ is in the list $\YY$
and in the path $\gamma$. In that case, both edges $(x,y)$ and $(x',z)$ are
removed from $\YY$.
\item If the list $\YY$ is empty, we are done.
\item Pick an edge $(x,y)\in\YY$.
\item Find the piece $\KK(n_{*,s_1,\dots,s_k})$, where $s_i\in\{\L,\R\}$,
among the leaves of the binary tree which contains the edge $(x,y)$.
We use the abbreviations $\ss=\{s_1,\dots,s_k\}$ and $n_{*,\ss}$.
The edge  $(x,y)$ belongs to exactly one piece\footnote{Note that the
  only edges which are common to more than one piece
are the edges of the
paths along which we already cut. Since $(x,y)$ is still in the list $\YY$, it
cannot
 be such an edge.}.
\item
Observe that the node $y$ is in
$\partial\II(n_{*,\ss})\cap \partial\II(n_*)$\footnote{By
  hypothesis, $y\in\partial\II(n_*)$ and therefore also
$y\in\partial\II(n_{*,\ss})\cap \partial\II(n_*)$.}.\hfill\break
$\bullet$ If $x$ is in the interior of $\KK(n_{*,\ss})$, the edge $(x,y)$
gives us the first simple path connecting $x$ to $\partial
\II(n_{*,\ss})$ and
by   \lref{l:3connected} there is a second independent path connecting $x$ to a
node $z\in\partial\KK(n_{*,\ss})$, $z\neq y$.

\ \ \ If $z$ is also in
$\partial\II(n_{*,\ss})$ we have found a $\gamma_\ss$ along
which we can cut. Note that in this case, the path $\gamma_\ss$ might promote a second node $x'$ of type C1; this 
happens if $z\in\partial \II(n_*)$ and $(x',z)$ is an edge of $\gamma_\ss$.

\ \ \ If $z\notin\partial\II(n_{*,\ss})$, the path $\gamma_\ss$ is obtained by
adding the edge which
connects $z$ to the tip of the cone\footnote{The distance between
  $\partial\II(n_{*,\ss})$ and any node in $\partial \KK(n_{*,\ss})$ is
  at most 1, see \fref{fig:newpaths}. The node $z$ belongs to a path $\gamma_{\S'}$ along which we already cut.
This implies that $z$ is connected to $n_{*,\S'\L}$ or $n_{*,\S'\R}$, called the tip of the cone associated with $z$.\label{footnote}}.\hfill\break
$\bullet$ If $x$ is not in the interior of $\KK(n_{*,\ss})$, $\gamma_\ss$
is found by connecting $x$ to a tip of one of the cones attached to
$\KK(n_{*,\ss})$\footnote{Note that the node $y$ is \emph{not} on a
  tip of a cone but is on the original boundary $\partial\II(n_*)$.}\,(see Footnote~\ref{footnote}).
\item We split $n_{*,\ss}$ along the path $\gamma_\ss$ and
add the 2 new pieces
  $\KK(n_{*,\ss\R})$ and $\KK(n_{*,\ss\L})$ to the tree
as leaves of $\KK(n_{*,\ss})$. Note that
  $\KK(n_{*,\ss})$ is no longer a leaf of the tree and will never be
encountered in the remaining steps of the algorithm.
  Finally, we remove the edge $(x,y)$ (and eventually $(x',z)$ if $x'$ is also promoted by $\gamma_\ss$) from the list
$\YY$.
\item We continue with step 3.
\end{myenum}

The algorithm stops when all
internal nodes of type C1 of $\II(n_*)$ have been promoted to C0.
Since each branch of the tree is used at most once and since we never cut along
the boundary of any $\KK_\ss$ we have shown:
\begin{theorem}\label{t:c1c0}
 The algorithm decomposes the triangulation $\II$ along a sequence
of simple paths; it promotes all of the internal nodes
 of $\II$ of type C1 into nodes of type $C0$. Furthermore,

every edge in $\II\setminus \partial\II$ is in at most
    one path.
\end{theorem}

\subsubsection{Change of the f-vector after the entire recursion}

In this section, we compute the total change in the f-vector resulting from
the elimination of all internal nodes.

Note that we will compute
the total increase in the number of internal edges instead of tetrahedra.
The two numbers are related by \eref{e:euler}.

\begin{definition}
  We need 3 counters at each depth $d$ of the original triangulation:
  \begin{myitem}
    \item $a_d$ is the number of internal edges $(x,y)$ with
      $D_x=d$ and $D_y=d+1$.
    \item $b_d$ is the number of internal edges $(x,y)$ with
      $D_x=d$ and $D_y=d$.
    \item $c_d$ is the number of internal faces $(x,y,z)$ with
      $D_x=d$ and $D_y=d+1$. (This implies $D_z=d$ or $d+1$.)
  \end{myitem}
\end{definition}

As every node is connected to nodes of the same depth or to depths
differing by at most 1, the following obvious relations hold:
\begin{equa}[e:sum]
 \sum _d (a_d + b_d) &= e~,\\
 \sum_d c_d &\le f_\i~,
\end{equa}
where $e$ is the number of internal edges, and $f_\i$ is the number of
internal faces.

Let $\Delta_d$ denote the increase of the number of internal edges
obtained when performing the steps C2$\to$C1$\to$C0$\to$external at level $d$.
\begin{proposition}\label{p:theprop}
  There is a constant $C'_\Delta$ such that
  \begin{equa}[e:Deltad]
    \Delta_d &\le C'_\Delta (a_d + b_d +a_{d-1} + c_{d-1})~,\text{ for }d>0~,\\
    \Delta_0 &\le C'_\Delta (a_d + b_d +n_\s)~,\text{ for }d=0~.
  \end{equa}
\end{proposition}

\begin{corollary}\label{c:increase}
  Eliminating all internal nodes of a triangulation $T$ with f-vector
  $\langle t,f_\s,n_\i \rangle$
leads to a
  total increase $\Delta$ of internal edges which is bounded by
  \begin{equ}
    \Delta \le C_\Delta (t +n_\i)~.
  \end{equ}
\end{corollary}

\begin{proof}[Proof of the corollary]
  From \eref{e:euler} we deduce $f_\i= 2t -f_s/2$ and
  $e=t+n_\i-f_\s/2+1$. Also, $n_\s=f_\s/2+2$.
 Using
  \eref{e:sum} and the proposition, we get
  \begin{equa}
    \Delta  &= \sum _{d\ge0} \Delta _d \le C'_\Delta (2e +
    f_\i+n_\s)~,\\
  \end{equa}
from which the assertion follows (the coefficient of $f_\s$ is negative and the
additive constants
can be bounded since $1\le t$).
\end{proof}
\begin{proof}[Proof of \pref{p:theprop}]
  The key to the bound
  \eref{e:Deltad} is the observation that the transformations
  C2$\to$C1$\dots$ are \emph{local} in the depth $d$ one
  considers. Indeed, as is visible from the definition of these
  transformations, working at level $d$ only affects $a_d$, $b_d$,
  $c_d$ and $a_{d-1}$, $c_{d-1}$.

More precisely, when starting to work at level $d$, we need the value of $\hat
a_{d-1}$, which is the number of internal edges (connecting depth
$d-1$ to $d$)
obtained when level $d-1$ has been completed.

As we work on level $d$, these values continue to change. After the
sweep C2$\to$C1 at level $d$ we obtain
$a'_d$, $\hat a'_{d-1}$, and similarly for the other variables. After
the sweep C1$\to$C0 we obtain $a''_d$ and other variables. The sweep
of removing the tetrahedra after C0 decreases all the counters, so we
do not introduce new notation.

The main bound is
\begin{lemma}\label{l:changead}
One has, after completing level $d-1$:
\begin{equ}\label{e:ad}
  \hat a_{d-1} \le a_{d-1} +16 c_{d-1}~.
\end{equ}
\end{lemma}
Postponing the proof, we recall the following facts:
\begin{myitem}
  \item Cutting along a path $\gamma$ adds $|\gamma|-1$ internal edges
    to the triangulation (see \lref{l:split}).
\item Each edge of each $\II(n_*)$ is used in at most 1 path
  $\tilde\gamma_\ss$
(see
  \tref{t:c2c1} and \ref{t:c1c0}).
\item The extension of the path $\tilde\gamma_\ss$ to a splitting path
  $\gamma_\ss$ adds at most 2 to its length (see step 7 for the case
  C2$\to$ C1, and step 6 for the case C1$\to$C0). We will use this
  observation by saying that $|\gamma_\ss| -1 \le 2 |\tilde\gamma_\ss|$.
\end{myitem}

Using these facts and \lref{l:euler2d}, the increase $\Delta_d'$ of
the number of internal edges due to the sweep
C2$\to$C1 is bounded by
\begin{equa}[e:deltaprim]
 \Delta' _{d} \le & 2\sum_{n_* \in \LL _d} \sum_\S |\tilde\gamma _{\S}|
 \le  2\sum_ {n_* \in \LL _d} \# \left( \text{edges in } \II(n_*) \right)\\
 \le & 6a_{d} + 12 b_d + 6 \hat a _{d-1} + 2\sum_{n_* \in \LL _d} (|\EE(n_*)|
-3 )~,
\end{equa}
Here, we over-count the number of added internal edges. However, one should keep in mind that if we follow the algorithms of 
Sections~\ref{s:c21} and \ref{s:c10}, then the Theorems~\ref{t:c2c1} and \ref{t:c1c0} are valid and 
every new internal edge is accounted for. As a consequence, the Relation~\eref{e:deltaprim} is an upper bound on the number of internal edges due to the sweep
C2$\to$C1.

The effect of the sweep  C2$\to$C1 at level $d$ is summarized by
\begin{lemma}\label{l:changec21}
 One has
\minilab{e:changec21}
\begin{equs}
 a'_d + b'_d + a'_{d-1} &\le a_d + b_d + \hat a_{d-1} + \Delta' _{d}~,
\label{e:changec21a}\\
 a'_d &\le a_d + 2c_d~, \label{e:changec21b}\\
 c'_d &\le 7c_d~.\label{e:changec21c}
\end{equs}

\end{lemma}

Postponing the proof, we proceed to the sweep C1$\to$C0. In the same manner,
the increase $\Delta''_d$ of internal edges for the
sweep C1$\to$C0 at level $d$ is
\begin{equa}[e:deltaprimprim]
 \Delta'' _{d} \le & 2\sum_{n_* \in \MM _d} \sum_\S |\tilde\gamma _{\S}|  \le
2\sum_ {n_* \in \MM _d} \# \left( \text{edges in } \II(n_*) \right)\\
 \le & 6a'_{d} + 12 b'_d + 6 a'_{d-1} + 2\sum_{n_* \in \MM _d}
(|\EE(n_*)|-3)~.
\end{equa}

To complete the proof of \pref{p:theprop}
we note that the external degree of $n_*$ is always 3 for those nodes
which have been promoted to the surface by removing a
tetrahedron. (Those which were at the surface at level $d=0$ can of
course have higher degree.) Using \eref{e:deltaprim} and  \lref{l:changead} we
get
\begin{equa}[e:deltad]
 \Delta' _{d}
 \le & 6a_{d} + 12 b_d + 6 \hat a_{d-1} + 4e_\s \cdot \delta_{d=0}  ~\\
 \le & 6a_{d} + 12 b_d + 6 a_{d-1} +96 c_{d-1}  + 12 n_\s \cdot \delta_{d=0}
{}~\\
 \le & 96(a_d + b_d + a_{d-1} + c_{d-1}+n_\s \cdot \delta_{d=0})~.
\end{equa}

In \eref{e:deltaprimprim} the external degree of a node $n_* \in \MM_d$ can be
larger than 3. However, if
we split a node $n_*$ into $n_{*,\R}$ and $n_{*,\L}$, then the
external degrees satisfy
\begin{equ}\label{e:extdegsplit}
 |\EE(n_*)| = |\EE(n_{*,\L})| +|\EE(n_{*,\R})| - 4~.
\end{equ}
Therefore, we can bound
\begin{equs}
  \sum_ {n_* \in \MM _d} (|\EE(n_*)| - 3) &\le 4\cdot \#\left( \text{splits in
C2$\to$C1}  \right)~.
\end{equs}
Since each split adds at least one internal edge, we deduce that
\begin{equs}
    \sum_ {n_* \in \MM _d} (|\EE(n_*)| - 3) &\le 4 \Delta' _{d} ~.
\end{equs}
Combining this with \eref{e:changec21a} and \lref{l:changead}, we get
\begin{equa}
 \Delta'' _{d} \le &  6a'_{d} + 12 b'_d + 6 a'_{d-1} + 2\sum_{n_* \in \MM _d}
(|\EE(n_*)|-3)\\
               \le &  6a_{d} + 12 b_d + 6 a_{d-1} + 96c_{d-1} + 20 \Delta'_d~.
\end{equa}
Replacing $\Delta'_d$ with \eref{e:deltad} yields the result we seek.
The last step C0$\to$external adds no internal edges; in fact it
reduces their number. This finishes the proof of \pref{p:theprop}.
\end{proof}

\begin{proof}[Proof of \lref{l:changec21}]
  In the sweep C2$\to$C1 at level $d$, we split all (or some) nodes $\{n_*\}
\subset \LL_d$ into $\{n_{*,\ss}\}\subset\MM_d$. As a consequence, all
$\Delta'_d$ added internal edges
 have an end $n_{*,\ss}$ which is the child of some node $n_*\in\LL_d$ at depth
$d$ in the initial triangulation. The number of internal edges having a corner
at depth $d$ is given by
 $a_d + a_{d-1}+b_d$. This proves the relation \eref{e:changec21a}.

To prove \eref{e:changec21b}, we need to bound the added number of internal
edges $(n_{*,\ss},y)$ such that $n_* \in \LL_d$ and $y\in\II(n_*)\cap\LL_{d+1}$
was at depth
$D_y=d+1$ in the original triangulation (and therefore is at depth 1 in the
current step).
By construction, this number is bounded by the number of paths $\gamma_\ss$
which go through such a node $y$ in the 2d triangulation
$\II(n_*)$. Furthermore, by \tref{t:c2c1}, each edge of $\II(n_*)$ is used in
at most one path $\gamma_\ss$. We deduce that, for two such nodes
$n_*$ and $y$, the number of added internal edges of type $(n_{*,\ss},y)$ is
bounded by the degree of the edge $(n_*,y)$ in the original triangulation.
Summing the degrees of all edges $(n_*,y)$ such that $n_* \in \LL_d$ and $y\in
\II(n_*) \cap \LL_{d+1}$ is bounded by $2c_d$.

Finally, in order to prove \eref{e:changec21c},  we need to bound the added
number of internal faces $(n_{*,\ss},y,z)$ in the step C2$\to$C1 at level $d$
when
$n_{*,\ss}$ is obtained from splitting some $n_* \in \LL_d$ and $y \in
\LL_{d+1}\cap\II(n_*)$. But each added internal face $(n_{*,\ss},y,z)$ requires
the addition of
the  internal edge $(n_{*,\ss},y)$.
Furthermore, by definition of the move split-a-node-along-a-path, each new
internal edge is added along with three internal faces. We deduce that $c'_d -
c_d \le 3(a'_d - a_d) \le 6c_d$.
\end{proof}

\begin{proof}[Proof of \lref{l:changead}]
The proof follows by induction on the level $d$. At level $d-1=0$,
there is nothing to prove (since no splits have been done). When we
are at level $d-1>0$ we can use \pref{p:theprop} (at level $d-1$).
 Following the same reasoning used in the proof of \eref{e:changec21b}, we can write
\begin{equ}
 a''_{d-1} \le a'_{d-1} + 2c'_{d-1}~.
\end{equ}
Replacing \eref{e:changec21b} and \eref{e:changec21c}, we get
\begin{equs}
 a''_{d-1} &\le a_{d-1} + 16c_{d-1}~.
\end{equs}
Finally, $\hat a_{d-1} \le a''_{d-1}$
since the third step C0$\to$external does not add internal edges (this third
step actually removes $3|\LL_{d-1}|$ internal edges).

\end{proof}

\subsection{Reducing a triangulation with no internal nodes into a
  set of nuclei}
Let $T$ be any triangulation. In the previous section, we described an
algorithm which transforms
$T$ into a new triangulation $T'$ with no internal nodes. We now systematically
apply the moves
cut-a-3-face and open-a-2-face on every internal face of $T'$ with less than 2
internal edges. We end
up with a collection of triangulations $\{N_i\}$ satisfying the
following properties:
\begin{myitem}
 \item All nodes of any such $N_i$ are external.
 \item All internal faces of any such $N_i$ have at least 2 internal edges.
\end{myitem}

Any triangulation satisfying these two conditions is called a \emx{nucleus}.

\section{Part II: Bounding the number of triangulations}\label{s:part2}
We showed that any triangulation can be reduced into a collection
of nuclei using four moves.
For the moment, we proceed without using the move cut-a-3-face. This implies
that
any triangulation can be transformed into a ``tree of nuclei'' (the
formal definition of a tree of nuclei will be given later on)
using the three remaining moves. Equivalently, this shows that any
triangulation can be
constructed from a tree of nuclei, using the inverse of these three
moves. Bounding the number of trees of nuclei, and then bounding the
number of ways one can perform the inverse moves
on such a tree yields a bound on the total number of triangulations.
\subsection{Rooted triangulations}
We define what we mean by a rooted triangulation $T$ and we show that one
can label all external nodes of $T$. In the sequel, we use a particular labeling described below.

\begin{definition}
 A \emx{rooted triangulation} $(T,F)$ of the 3-ball is a triangulation $T$ with one labeled
external face $F$. This labeled face is called the \emx{root}. The three nodes of
the root
 are always labeled 0, 1, and 2.
\end{definition}

  \begin{figure}[h]
    \begin{center}
      \includegraphics[width=0.5\textwidth]{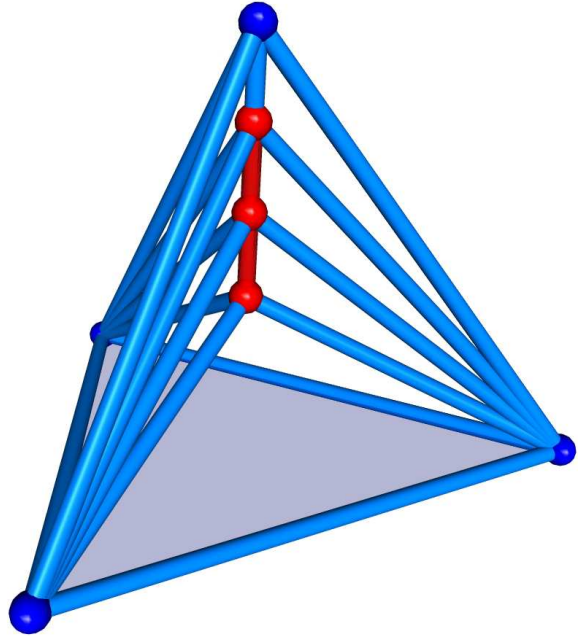}
      \caption{The Christmas tree with $m=3$ internal nodes. This triangulation
can be rooted in more than one way.}
      \label{fig:christmas}
    \end{center}
  \end{figure}

\begin{remark}
 We will only consider rooted triangulations. This means for instance that
talking about the Christmas tree $T_m, m>1$
makes no sense, since there is
more than one such
 rooted triangulation. The exceptions
are of course symmetric triangulations $T$ such as the tetrahedron.
\end{remark}

\begin{proposition}\label{p:labeling}
 Consider the boundary of a rooted triangulation $(T,F)$. The root
 is labeled as $(0,1,2)$. One can define a way of labeling
all external nodes of $T$.
\end{proposition}
\begin{proof}
The proof is just the construction of this labeling.
Any labeled edge can be seen as an element $(a,b) \in \mathbb{Z}_+^2$ with
$a<b$.\footnote{We use the notation $\mathbb{Z}_+=\{0,1,2,\dots\}$.} We
consider the lexical order on $\mathbb{Z}_+^2$.
We start with the node 0. Its external flower is a 1d triangulation of the
circle $S^1$ and it contains the edge $(1,2)$ by definition. This edge
determines the direction
 in which we label all unlabeled nodes of the external flower of node 0.

Next, we consider the external flower of node 1 and we look for the smallest
labeled edge in the sense of the above ordering. In this case, this edge is
$(0,2)$. This edge fixes the direction
in which we label all the yet  unlabeled nodes of the
 external flower of node number 1.
 Notice that all unlabeled nodes which are assigned a label are part of a face
along with 2 already labeled nodes.
 This implies that the external flower of any labeled node contains a
smallest labeled edge and as such can be directed.

 We continue with all the nodes in their natural order until all
 external nodes of $T$ are labeled.
\end{proof}

\subsection{Trees of nuclei}

Since we work with rooted triangulations, from now on, we will only use rooted nuclei, namely:
\begin{definition}
 A \emx{nucleus} is a \textbf{rooted} triangulation with no internal nodes
such that every internal face has at most one external edge.
\end{definition}

\subsubsection{Rooted trees of nuclei and planar rooted trees}

Let $\mathcal{N}$ be the set of all nuclei and $\mathcal{N}_{t,f}$ be
the subset
of all nuclei with $t$ tetrahedra and $f$ external faces.
\begin{definition}
 A rooted triangulation $T$ is called a \emx{rooted tree of nuclei} if all
nodes of $T$ are external and all internal faces of $T$ have 0, 2, or
3 internal edges. (In other words, no internal face has 2 external edges.)
\end{definition}

In other words, a rooted tree of nuclei is simply a rooted triangulation
which is obtained by gluing sequentially nuclei along pairs of their external
faces.
This is done in such a way that each nucleus is glued to an external face $(a,b,c)$ of
its parent through its
root; 0 is identified with $a$, 1 with $b$ and 2 with $c$.
Once the tree is built, the external nodes are renumbered in the sense
of \pref{p:labeling}.

Since all external faces of a rooted triangulation are ordered, this defines
a bijection between rooted trees of nuclei $(T,F)$ and rooted planar trees with
colored vertices in the following manner:
\begin{myitem}
 \item Each nucleus of the triangulation $(T,F)$ is represented by a colored
vertex.
 \item The root-vertex of the planar tree represents the nucleus with the root
$F$, \ie,  with the face $(0,1,2)$.
 \item Each internal face of the triangulation with three external edges is
shared by two nuclei and hence it is represented in the tree by an
edge linking the
corresponding two colored vertices.
 \item Since the internal faces with three external edges are ordered,
this induces an order of the
links of the planar tree, say from left to right.

\end{myitem}
\subsubsection{Hypothesis on the number of rooted nuclei}

We next show how the question of Gromov can be reformulated. We show
that if there are not ``too many'' different types of nuclei, then
there is indeed an exponential bound on the number of triangulations,
when expressed in terms of the number of tetrahedra.

\begin{hypothesis}\label{h:assumption}
There is a finite constant $K_1>1$ such that the number $\rho(t,f_\s)$ of
face-rooted nuclei  with f-vector $\langle t,f_\s,0\rangle$ is bounded by $K_1^{t}$.
\end{hypothesis}

In order to alleviate the notation, from now on, we will denote $f_\s$ by $f$.

\begin{lemma}\label{l:nucleus}
For any nucleus $N\in\NN_{t,f}$ one has $f\le t+3$.
\end{lemma}
\begin{proof}
  If $N$ is a tetrahedron, the assertion is obvious.
If $N$ is non-trivial each tetrahedron of $N$ can have at most
1 external face,
since otherwise it would have an internal face with more than one external
edge.
\end{proof}

\subsubsection{The number of rooted trees of nuclei}

We use the classical method for counting planar ordered trees,
generalized to the case of a multitude of different nodes, which are
the face-rooted nuclei.
\begin{definition}

Let $A_{v,t,f}$ be the number of rooted trees of nuclei with $v>0$ nuclei, $t$
tetrahedra and $f$ external faces. We define $A_{0,t,f} = \delta
_{t,0}\,\delta _{f,0}$.

\end{definition}

Our main bound is:
\begin{proposition}\label{p:k2}Under the Hypothesis \ref{h:assumption}
  there is a $K_2$, with $2<K_2<\infty$ such that for all $t,f$, one has
  \begin{equ}
    \sum_v A_{v,t,f}\le K_2^t~.
  \end{equ}
\end{proposition}

\begin{proof}

Consider a tree of nuclei, and
let $N$ be the nucleus containing the root $F$ and assume that
$N\in\mathcal{N}_{t_0,f_0}$.
Removing $N$ from the tree leads to $f_0-1$ rooted
trees of nuclei, some of which may be empty. We let $v_i$, $t_i$, and
$f_i$ denote the counters for the branch $i$.
Note that if a branch $i$ has 0 nuclei, \ie,
if $v_{i} = 0$, then, obviously,  $t_{i} = f_{i} = 0$.
Thus, we get the relations:
\begin{equa}\label{e:newrelation}
\sum_{i=1}^\ell v_i = v-1~,\qquad
\sum_{i=1}^\ell t_i &= t-t_0 ~,\qquad
 \sum _{i=1}^{f_0-1} \delta_{f_i > 0}(f_i-1) + \delta_{f_i=0} = f-1~.
\end{equa}
In the sequel, we denote by $\sum'_{v,t,f,t_0,f_0}$ the sum over the set
\begin{equ}
  \bigl\{
v_i,t_i,f_i~ \big | ~i=1,\dots,f_0-1~, v_i\ge 0, t_i\ge 0, f_i\ge0
\text{~and satisfying~\eref{e:newrelation}}
\bigr\}~.
\end{equ}
This observation allows us to write a recursive relation
\begin{equa}\label{e:encore2}
 A_{v,t,f} &= \delta _{v,0}\,\delta
_{t,0}\,\delta _{f,0}+ \sum_{t_0>0,f_0 \ge 4} \rho(t_0,f_0) ~
{\sum}'_{v,t,f,t_0,f_0}\  \prod_{i=1}^{f_0-1}A_{v_i,t_i,f_i}~.
\end{equa}
Fix $M \in \mathbb{Z}_+$, and assume that $v,t,f$ satisfy $3v+3t+f \le M$.
By \eref{e:newrelation}, we deduce
\begin{equ}
 3v_i+3t_i+f_i \le 3v-3 + 3t-3t_0 + f \le M-1 ~.
\end{equ}
We define
\begin{equ}
 A_{M} (s) = \sum_{3v+3t+f \le M} A_{v,t,f}s^{3v+3t+f}~.
\end{equ}

Clearly, $A_0(s) = 1$ for all $s$, $A_M(0)=1$ for all $M\ge0$, and for
a fixed $s$, $A_M(s)$ is an increasing sequence in $M$.

Multiplying \eref{e:encore2} by $s^{3v+3t+f}$ and summing, we get,
using \eref{e:newrelation}:
\begin{equs}[e:eqtrees]
 A_M(s) &= 1 + \sum_{3v+3t+f\le M}~\sum_{t_0 = 1} ^{t} \sum_{f_0 = 4} ^{f}
\rho(t_0,f_0) ~s^{3+3t_0+1-\sum _{i=1} ^{f_0-1} (\delta_{f_i > 0} -
\delta_{f_i=0})}\\
&\qquad\qquad
\times {\sum}'_{v,t,f,t_0,f_0}
\prod_{i=1}^{\ell}A_{v_i,t_i,f_i}s^{3v_i+3t_i+f_i}~.\\
\end{equs}
Using \lref{l:nucleus}, we have
\begin{equa}
  3+3t_0+1-\sum _{i=1} ^{f_0-1} (\delta_{f_i > 0} -
  \delta_{f_i=0})&\ge
 3+3t_0+1-(f_0-1)\cdot1 +0\\
&\ge 5+3t_0-f_0=2(t_0+3-f_0)+t_0+f_0-1\\
&\ge t_0+f_0-1~.
\end{equa}
Restricting to  $0\le s\le1$, this implies
\begin{equ}\label{e:encore1}
s^{3+3t_0+1-\sum _{i=1} ^{f_0-1} (\delta_{f_i > 0} - \delta_{f_i=0})} \le
s^{t_0 + f_0-1}~.
\end{equ}
Using now the Hypothesis \ref{h:assumption}, \ie, $\rho(t,f) \le K_1^{t}$, we
get from
\eref{e:eqtrees} and \eref{e:encore1}:
\begin{equa}
 A_M(s) -A_M(0) &\le  \sum_{t_0 = 0} ^{M} (sK_1)^{t_0} \sum_{f_0-1 = 0} ^{M}  ~
\prod_{i=1}^{f_0-1}sA_{M-1}(s)
        \le  \frac{1-(sK_1)^{M+1}}{1-sK_1}
\frac{1-\left(sA_{M-1}(s)\right)^{M+1}}{1-\left(sA_{M-1}(s)\right)}~.\\
\end{equa}
Restricting  $s$ further to $s\le 1/(2K_1)$ this leads to
\begin{equa}
 A_M(s) -A_M(0) &\le
2\frac{1-\left(sA_{M-1}(s)\right)^{M+1}}{1-\left(sA_{M-1}(s)\right)}~.\\
\end{equa}
Fix $s^* = \min (0.1,1/(2K_1))$ and consider the map $F: x\mapsto
1+2/(1-s^*\cdot x)$. One easily checks
that $F$ maps the interval $[1,5]$ to itself. Furthermore, we have
$s^*\cdot x \le 1$ for $x\in [1,5]$.
Starting with $x=A_0(s^*)=1$ we conclude that for all $M$
one has $A_M(s^*) \le 5$. This implies that the monotone sequence $A_M(s^*)$
converges as $M\to\infty$ and thus
\begin{equ}
  A_{v,t,f} \le 5\cdot (s^*)^{-3v-3t-f}~.
\end{equ}
Summing over $v$ and using $v\le t$ and $f\le 4t$ we complete the proof.
\end{proof}

\subsection{Bound on triangulations}

Having discussed the number of trees, we now study the number of ways
these trees can be made into triangulations by identifying
faces and nodes. This process is patterned after the work of \cite{DJ1995} and
\cite{BZ2011}.

Our bounds are based on using the inverses of the moves
  open-a-2-face, remove-1-tetra, and split-a-node-along-a-path.
Since we are only
  interested in the bound, we will allow for inverse moves which do
  not necessarily lead to 3-balls.

  \begin{remark}
    While we over-count the number of triangulations, by allowing for
    moves which may not lead to 3-balls, we can in fact formulate
    precise conditions which guarantee that after each move, a  3-ball
    is obtained. These conditions are spelled out in Lemmas
    \ref{l:identify} and \ref{p:contract}. This actually allows for
    efficient programming of the inverse operations.
  \end{remark}

\subsubsection{Bounding the number of rooted triangulations with no internal
nodes}

Let $\mathcal{R}_{t,f}$ be the set of all rooted trees of nuclei with $t$
tetrahedra and $f$ external faces and let $\mathcal{T}_{t,f,0}$ be the set of
all rooted triangulations
with $t$ tetrahedra, $f$ external faces and no internal nodes.
In this section, we will define the inverse move of open-a-2-face and we will
use it
to count the number of rooted triangulations with no internal nodes.

The inverse operation of open-a-face, which we will simply call
\emx{identification} when there is no ambiguity, is to identify two adjacent
external faces,
satisfying some conditions. Indeed, identifying any two adjacent
external faces might lead to a complex which is not a triangulation. For
instance, assume that $(n_1,n_2,m_1)$ and $(n_1,n_2,m_2)$ are two adjacent
external faces
such that there exists a node $x$ adjacent to both $m_1$ and $m_2$. After
identifying the two faces, we obtain a complex with a double edge $(x,m_1) =
(x,m_2)$.

\begin{lemma}\label{l:identify}
 Consider a triangulation $T$. Let $(a,b)$ be an external edge and let $x,y$ be
its two opposite external nodes. Assume that the following conditions are
satisfied:
   \begin{myitem}
     \item The nodes $x$ and $y$ are not connected by an edge.
     \item The only nodes $m$ such that $(m,x)$ and $(m,y)$ are edges are the
two nodes $a$ and $b$.
   \end{myitem}
 Then, one can identify the two external nodes $x$ and $y$ as well as
 the two external faces sharing $(x,y)$. This operation
 transforms a 3-ball to a 3-ball, and will be called identification (of two adjacent external faces).
\end{lemma}
\begin{proof}
  The proof is left to the reader.
\end{proof}

\begin{proposition}\label{p:k3}
  Under Hypothesis \ref{h:assumption}, there is a constant $K_3$
  such that for all $t$ and $f$ one has
  \begin{equ}
    |\TT_{t,f,0}| \le K_3^t~.
  \end{equ}
\end{proposition}
\begin{proof}
Let $T\in\mathcal{T}_{t,f,0}$ be any rooted triangulation with no internal
nodes. Using repetitively the move open-a-2-face on $T$ transforms it into a
rooted triangulation $T'$ with no
internal nodes such that each internal face has 0, 1 or 3 external
edges. In other words, $T'$ is a rooted tree of nuclei.
Equivalently, given a rooted tree of nuclei $T'$ with $t'$ tetrahedra
and $f'$ external faces, one can count the number of ways one can
identify two adjacent external faces, \emph{without any conditions
guaranteeing ballness}. Multiplying this number by the number of rooted trees
of nuclei gives us an upper bound on the number of rooted triangulations
with no internal nodes.

We count the number of $T\in \mathcal{T}_{t,f,0}$ obtained by
identification from a
rooted tree of nuclei $T'$ with $t'$ tetrahedra and $f'$ external faces. This
means that we identify $D = (f'-f)/2$
pairs of adjacent external faces.

We first observe that choosing a pair of adjacent
external faces is equivalent to choosing an external edge.
We then note that some faces which are not adjacent in $T'$
might become adjacent after some identifications are done.
This means that we have a sequence $e_1,e_2,\dots,e_\ell$ with $e_i\ge1$ and
$\sum_i e_i = D$ which is defined as follows:
\begin{myitem}
 \item $e_1$ is the number of external edges (or equivalently of pairs of
adjacent external faces) of $T'$ which are identified.
 \item $e_2$ is the number of pairs of faces which were not adjacent in $T'$
but became so after the first series of $e_1$ identifications. However, each
identification of
 two adjacent external faces creates exactly two new pairs of adjacent
external faces, implying that $e_2 \le 2e_1$.
 \item $e_i$ is defined by analogy from the $e_{i-1}$ identifications,
implying
 that $e_i \le 2e_{i-1}$.
\end{myitem}
This leads to the following bound:
 \begin{equs}
  |\mathcal{T}_{t,f,0}| &\le \sum_{f'>f} |\mathcal{R}_{t,f'}| \sum_{\ell =
1}^{D\equiv (f'-f)/2} ~\sum_{\sum_{i=1}^\ell e_i = D , e_i\ge1}~
\binom{3D}{e_1}\binom{2e_1}{e_2}\dots\binom{2e_{\ell-1}}{e_\ell}~.
 \end{equs}
Since $\binom{a}{b} \le 2^a$, and since the number of external faces $f'$
in any rooted
tree of nuclei is bounded by four times the number of tetrahedra, we find, using \pref{p:k2} to bound 
$|\mathcal{R}_{t,f'}|$,
 \begin{equs}
  |\mathcal{T}_{t,f,0}|
                        &\le \sum_{f'>f} |\mathcal{R}_{t,f'}| 2^{5(f'-f)/2}
\sum_{\ell = 1}^{D\equiv (f'-f)/2} ~\sum_{\sum_{i=1}^\ell e_i = D ,
  e_i\ge1}\kern -1em 1 \\
                        &\le \sum_{f'>f} |\mathcal{R}_{t,f'}| 2^{5(f'-f)/2}
\sum_{\ell = 1}^{D\equiv (f'-f)/2} \binom{D-1}{\ell-1} \\
                        &\le \sum_{f'>f} |\mathcal{R}_{t,f'}| 2^{3(f'-f)} \\
                        &\le \sum_{f'=f+2}^{4t} K_2 ^{t} K_2^{3(f'-f)} \\
                        &\le K_2^{13t}= K_3^t~,
 \end{equs}
where $K_3 = K_2^{13}$.

The proof is complete.
\end{proof}

\subsubsection{Bounding the number of rooted triangulations
  (internal nodes included)}
In this section, we define the inverse moves of remove-1-tetra and
split-a-node-along-a-path and we use them to count the number of
rooted triangulations.

\begin{definition}
 We define the inverse move of remove-1-tetra, which we call
 \emx{adding a tetrahedron}\emph{:}
 Consider a triangulation $T$. Let $x$ be an external node with external degree
equal to 3 and let $a_1$, $a_2$ and $a_3$ be its external neighbors, \ie,
$(x,a_i)$ is
 an external edge. Adding a tetrahedron then consists in adding  the face
$(a_1,a_2,a_3)$ and the tetrahedron $(x,a_1,a_2,a_3)$.
\end{definition}

 We define the inverse move of split-a-node-along-a-path.

\begin{lemma}
 Consider a triangulation $T$. Let $(a,b)$ be an external edge. Assume that the
following conditions are satisfied:
   \begin{myitem}
     \item For each node $m$ such that $(m,a)$ and $(m,b)$ are edges, $(m,a,b)$
is a face.
     \item For each edge $e$ such that $(e,a)$ and $(e,b)$ are faces, $(e,a,b)$
is a tetrahedron.
     \item There are no faces $f$ such that $(f,a)$ and $(f,b)$ are both
tetrahedra.
   \end{myitem}
 Then, one can collapse the two nodes $a$ and $b$, and the result is
 again a 3-ball.  This move is called
 \emph{collapse of an external edge} or simply \emph{collapse}.
\end{lemma}
\begin{proof}
  The proof is left to the reader.
\end{proof}

The three conditions of a collapse can be reformulated in the following manner:

\begin{lemma}\label{p:contract}
 Let $e = (a,b)$ be an external edge. The edge $e$ is collapsible if and only
if
\begin{equ}
 \mathcal{I}(a)\cap\mathcal{I}(b)=\mathcal{I}(e)~,
\end{equ}
where $\mathcal{I}(a)$ is the hemisphere of $a$ and $\mathcal{I}(e)$ is the
semi-circular flower of $e$.
\end{lemma}
\begin{proof}
 By definition, an edge $e=(a,b)$ is collapsible if and only if
 \begin{myitem}
     \item For each node $m$ such that $(m,a)$ and $(m,b)$ are edges, $(m,a,b)$
is a face.
     \item For each edge $e$ such that $(e,a)$ and $(e,b)$ are faces, $(e,a,b)$
is a tetrahedron.
     \item There are no faces $f$ such that $(f,a)$ and $(f,b)$ two tetrahedra.
   \end{myitem}

 Any graph is defined as a set of vertices and a set of edges. A 2d
triangulation is a graph that can be defined as a set of nodes, a set of edges
and a set of faces,
 and a 1d triangulation as a set of nodes and a set of edges. $\mathcal{I}(a)$
is a 2d triangulation and $\mathcal{I}(e)$ is a 1d triangulation.
 Let $\mathcal{V}(a)$, $\mathcal{L}(a)$ and $\mathcal{F}(a)$ be the sets of
vertices, edges and faces of $\mathcal{I}(a)$ and $\mathcal{V}(e)$,
$\mathcal{L}(e)$ those of
 $\mathcal{I}(e)$. The proposition is equivalent to the following
 \begin{equs}
  \mathcal{V}(a) \cap \mathcal{V}(b) &= \mathcal{V}(e)~, \\
  \mathcal{L}(a) \cap \mathcal{L}(b) &= \mathcal{L}(e)~, \\
  \mathcal{F}(a) \cap \mathcal{F}(b) &= \emptyset~.
 \end{equs}
The two definitions are clearly equivalent.
\end{proof}

In \sref{s:intnodes}, we described an algorithm which transforms any
triangulation with f-vector $\langle t,f,n \rangle$ into a triangulation
with f-vector $\langle t',f',0 \rangle$. We have the following lemma:
\begin{lemma}\label{l:k4}
 There is a constant $K_4>0$ such that the f-vectors $\langle t,f,n \rangle$ and $\langle t',f',0 \rangle$
satisfy the following linear relation:
\begin{equ}\label{e:linbound}
t'\le K_4 t~,\qquad f'\le K_4t~,
\end{equ}
\end{lemma}
\begin{proof}
Let $e,e'$ be the number of internal edges of both triangulations. 
 By \cref{c:increase}, we have $e'-e\le C_\Delta (t+n_\i)$.
 Using \eref{e:euler} and $f_\s,n_\i \le 4t$, the result follows.
\end{proof}

This proves that any triangulation in $\TT_{t,f,n}$ can be obtained from a
triangulation with no internal nodes in $\TT _{t',f',0}$ with
a series of carefully chosen collapses and additions of tetrahedra, with
$t,f,n,t',f'$ satisfying \eref{e:linbound}.

We can now use a similar approach to that of the previous section.
It is clear that choosing a triplet of external faces for the move add-1-tetrahedron
is equivalent to choosing an external node $x$, and that
choosing a couple of external nodes for collapse is equivalent to choosing an
external edge.

\subsection{Combining the bounds}

Before we state our main result, we recall the
\begin{oneshot}
There is a finite constant $K_1>1$ such that the number $\rho(t,f)$ of
face-rooted nuclei  with f-vector $\langle t,f\rangle$ is bounded by $K_1^{t}$.
\end{oneshot}

\begin{theorem}\label{t:main}
  Under Hypothesis \ref{h:assumption} one has the bound: There is
  a finite constant $C$ such that the
  number of rooted triangulations with f-vector $\langle t,f,n\rangle$
  is bounded by
  \begin{equ}\label{e:thebound}
    |\TT_{t,f,n}| \le C ^{t}~.
  \end{equ}
\end{theorem}

\begin{proof}

Consider a rooted triangulation $T\in\mathcal{T}_{t,f,n}$ with $t$ tetrahedra,
$f$ external faces and $n$ internal nodes.
We showed that $T$ can be obtained from a rooted triangulation
$T'\in\mathcal{T}_{t',f',0}$ by a series of carefully chosen collapses and
additions of tetrahedra.

Note that the algorithm of \sref{s:intnodes} which transforms $T$ into
$T'$ can always be stopped when the last internal node of $T$ is
removed.
This implies that, in the inverse construction we are doing now, we
must start by adding tetrahedra to $T'$, and not by collapsing
external edges.
So the first step is to choose $n_1$ external nodes (of external
degree 3) out of the $f'/2+2$ external nodes of $T'$, and to insert a
tetrahedron on each of them with one tip at the node. We call this
``covering the node''.

This reduces the number of external edges from $3f'/2$ to $3(f'/2-n_1)$.
Then, we choose $m_1$ external edges and we collapse them.

\begin{remark}
 Any labeled triangulation is simply defined by the list of its tetrahedra
$\LL_t$. In this point of view, collapsing an external edge $e$ is
 simply the operation where we remove from $\LL_t$ all the tetrahedra of
$\EE(e)$.
 Let $e_1$ and $e_2$ be two collapsible edges. The construction
 implies that the order in which we collapse them is irrelevant and
 so, the idea that
 we simultaneously collapse $m_1$ edges makes sense.

One should pay attention to the case where we collapse two edges $e_1=(a,b_1)$ and $e_2=(a,b_2)$ such that 
$(b_1,b_2)=e_3$ is an edge. In this case, all tetrahedra sharing one of the three edges are removed.
Clearly, this yields the same result regardless of the order in which we collapse $e_1$ and $e_2$.
\end{remark}

The next step is to choose $n_2$ external nodes among the new
possibilities which appear after performing the first series of
coverings and collapses,
and cover them.
For each external edge $e$, we can associate four nodes: the two endpoints of
$e$ and the two nodes $x_1,x_2$ such that $(x_i,e)$ is an external face.
Assume that $x$ is one of the $n_2$ chosen external nodes.
The fact that $x$ appeared after the first series implies that $x$ is
either one of the four nodes associated with one of the
$m_1$ collapsed edges (note that these four nodes become three after the
collapse),
or that there is a node $y$ among the first $n_1$ nodes such that $(x,y)$ was
an external edge (before covering $y$ with a tetrahedron).
But each such $y$ has exactly 3 external neighbors. This implies that
$n_2 \le 3m_1 + 3n_1$ and
the number of ways to choose these nodes is bounded by $$\binom
{3(m_1+n_1)}{n_2}~.$$
Continuing in this way, we choose $m_2$ external edges and we collapse them.
Let
$e$ be such an edge. Again, $e$ was not among the first $m_1$ edges.
This implies that there must be a node $x$ of the series of $n_2$ covered
external nodes such that $(e,x)$ formed an external face before
covering $x$ with a tetrahedron. But for each such $x$ there are
exactly three external edges
satisfying this condition. We deduce that $m_2\le 3n_2$.

We continue adding tetrahedra and collapsing edges. This leads to two
sequences $n_i,m_i,i=1,\dots,\ell$, with $\ell\le n$, satisfying:
\begin{equs}[e:last]
1 \le n_i~,\qquad  0&\le m_i   \le 3n_i~,\qquad
 \sum_{i=1}^\ell n_i = n~, \\
 1 \le n_{i} &\le 3n_{i-1} + 3m_{i-1}~, \quad i>1~, \\
 \sum _{i=1} ^\ell& 2n_i + 2m_i + f = f'~.
\end{equs}
Note that some, or all, of the $m_i$'s might be equal to zero.
Using \eref{e:last} we get a bound
 \begin{equs}
  |\mathcal{T}_{t,f,n}| \le \sum_{t',f'} |\mathcal{T}_{t',f',0}|
  &\sum_{\ell =1}^{n} ~\sum_{\sum_{i=1}^\ell n_i= n ,
    n_i\ge1}~\sum_{\sum_{i=1}^\ell m_i = (f'-f)/2-n , m_i\ge0}~\\
&\times\binom{f'/2+2}{n_1}\binom{3(n_1+m_1)}{n_2}
\cdots\binom{3(n_{\ell-1}+m_{\ell-1})}{n_\ell} \\
&\times\binom{3f'/2}{m_1}\binom{3n_1}{m_2}\cdots\binom{3n_{\ell-1}}{m_\ell}~,
 \end{equs}
where the sum over $t',f'$ is restricted by \eref{e:linbound}.
Bounding each binomial by a power of 2 and using \pref{p:k3},
\eref{e:last} and \eref{e:linbound}, we get, as in the proof of \pref{p:k3},
 \begin{equs}
    |\mathcal{T}_{t,f,n}| &\le \sum_{t',f'\le K_4t}  K_3^{t'} 
 \le C^t~.
 \end{equs}
This shows \eref{e:thebound} and completes the proof.

\end{proof}
\clearpage

\begin{acknowledge}
We profited from useful discussions with G.~Ziegler. This work was
partially supported by the Fonds National Suisse and by an ERC FP7 ``Ideas''
Advanced Grant.
\end{acknowledge}
\bibliographystyle{JPE}
\addcontentsline{toc}{section}{References}
\bibliography{paper}

\end{document}